\documentclass{article}
\usepackage{arxiv}
\usepackage{amsmath,amssymb,amsthm,mathtools,bm,amsfonts}
\usepackage{comment}
\usepackage{float}
\usepackage{algorithm}
\usepackage{algpseudocode}
\usepackage{booktabs}
\usepackage{subcaption}
\usepackage{graphicx}
\usepackage{setspace}
\usepackage[hidelinks]{hyperref}
\usepackage{url}
\usepackage{xcolor} 
\usepackage{doi}
\usepackage{bm}
\usepackage{bbm}
\usepackage[utf8]{inputenc}
\usepackage[T1]{fontenc}
\usepackage{microtype}
\usepackage{float}

\newtheorem{theorem}{Theorem}
\newtheorem{lemma}{Lemma}
\newtheorem{proposition}{Proposition}

\newtheorem{assumption}{Assumption}

\newtheorem{cond}{Condition}
\theoremstyle{remark}

\usepackage[
  backend=biber,      
  style=authoryear,   
  sorting=nyt         
]{biblatex}
\newcommand{\R}{\mathbb{R}}
\newcommand{\E}{\mathbb{E}}

\newcommand{\K}{\mathbb{K}}

\newcommand{\1}{\mathbb{I}}

\title{Semiparametric robust estimation of population location}

\author{
{Ananyabrata Barua}\thanks{Pre-Doctoral research scholar supported by Walmart Centre for Tech Excellence.\\Formerly at Indian Statistical Institute, Kolkata}\\
Department of Computer Science and Automation\\
Indian Institute of Science\\
Bengaluru, India\\
\texttt{ananyabratab@iisc.ac.in}
\And
{Ayanendranath Basu}\\
Interdisciplinary Statistical Research Unit\\
Indian Statistical Institute\\
Kolkata, India\\
\texttt{ayanbasu@isical.ac.in}
}

\date{} 

\renewcommand{\shorttitle}{\textit{Semiparametric robust estimation of population location}}

\hypersetup{
pdftitle={Semiparametric Robust Estimation of Population Location},
}

\begin{document}
\maketitle

\begin{abstract}
Real‑world measurements often comprise a dominant signal contaminated by a noisy background. Robustly estimating the dominant signal in practice has been a fundamental statistical problem. Classically, mixture models have been used to cluster the heterogeneous population into homogeneous components. Modeling such data with fully parametric models risks bias under misspecification, while fully nonparametric approaches can dissipate power and computational resources. We propose a middle path: a \emph{semiparametric} method that models only the dominant component parametrically and leaves the background completely nonparametric, yet remains computationally scalable and statistically robust. So instead of outlier downweighting, traditionally done in robust statistics literature, we maximize the observed likelihood such that the noisy background is absorbed by the nonparametric component.
Computationally, we propose a new FFT-accelerated approximate likelihood maximization algorithm. Empirically, this FFT plug‑in achieves order‑of‑magnitude speedups over vanilla weighted EM while preserving statistical accuracy and large sample properties.

\end{abstract}

\section{Introduction}
Finite mixture models have been a standard tool for modeling heterogeneous populations, discovering subpopulations, and performing model‑based clustering. Classical formulations impose a common parametric family on all components (e.g., Gaussian mixtures) and rely on identifiability assumptions or parsimonious parameterizations to render the estimation problem well‑posed \parencite{Bordes2006, ArcidiaconoJones2003, Dasgupta1999, DebSahaGuntuboyinaSen2022}. The EM algorithm is the workhorse in practice \parencite{Titterington1985,McLachlanPeel2000,Dempster1977,FraleyRaftery2002}. Yet the assumptions that make parametric mixtures tractable—common family, structured covariances, or a fixed number of components—are almost never aligned with actual noisy datasets. In many applications the investigator primarily cares about the \emph{most dominant signal peak} (its location and mass), while the remainder of the distribution is a complicated amalgam of weaker modes, tails, or background phenomena that defy clear parametric specification. For such tasks, fitting a fully parametric $K$‑component model can be brittle: identifiability hinges on modeling \emph{all} components correctly, and estimation can deteriorate due to over‑fitting or weak identifiability \parencite{Teicher1963,Chen1995}.

A parallel line of work relaxes parametric assumptions on \emph{one} component. In two‑component mixtures where one distribution is known and the other is nonparametric (sometimes with symmetry or shape constraints), there are strong identifiability and estimation results for the \emph{mixing proportion} and for features of the unknown component; see, e.g., \parencite{Bordes2006} and \parencite{patra2016estimation}. In machine learning, the related \emph{mixture‑proportion estimation} problem has been studied in semi‑supervised novelty detection and positive‑unlabeled learning, with estimators and rates under separability or irreducibility conditions \parencite{BlanchardLeeScott2010,Scott2015,RamaswamyScottTewari2016}. These strands underscore a theme: if the inferential target is a \emph{weight} or a \emph{single component}, strict distributional modeling of the remainder can be avoided.

Our paper embraces this philosophy and pushes it in a direction that is \emph{distinct from classical robust statistics}. Robust methods \`a la Huber and Hampel downweight or reweight observations to resist the effect of outliers—with M‑estimators, influence‑function design, or divergence‑based criteria—but still typically posit a parametric data‑generating model and optimize a single loss over all data. Robust clustering extensions either replace Gaussian errors by heavier‑tailed $t$ components or trim a fraction of points to curb spurious solutions. While powerful for outlier resistance, these approaches differ in motivation and technique from our objective of \emph{extracting the dominant peak and its mass without modeling the rest} \parencite{Huber1964,Hampel1986,Basu1998}.

We consider a semiparametric “contamination’’ mixture
$
f(y)\;=\;\pi_0\,f_0(y;\theta)+(1-\pi_0)\,f_{\mathrm{bg}}(y),
\qquad \pi_0\in(0,1),
$
in which $f_0(\cdot;\theta)$ is a \emph{parametric family} for a dominant component (e.g., a location or location–scale family) and the background density $f_{\mathrm{bg}}$ is left \emph{completely unrestricted}. The inferential target is the dominant component’s location parameter $\theta$. We do \emph{not} attempt full identifiability of a multi‑component parametric mixture; rather, our goal is targeted recovery of the dominant component under certain regularity. This formulation sits at the intersection of semiparametric mixtures “with one unknown component’’ and mixture‑proportion estimation \parencite{Bordes2006,BlanchardLeeScott2010,Scott2015,patra2016estimation,RamaswamyScottTewari2016}.

Algorithmically, we maximize a \emph{plug–in} log‑likelihood where the unknown background $f_{\mathrm{bg}}$ is replaced by a kernel density estimator (KDE) computed once per EM iteration via an FFT‑accelerated convolution on a uniform grid with linear binning and zero padding \parencite{Silverman1982,Silverman1986,Wand1994,WandJones1995}. This reduces the nonparametric update from $O(nG)$ to $O(G\log G)$ for $G$ grid points, while off‑grid evaluations use linear interpolation whose error we bound by the usual $C^2$ truncation rate $O(\Delta^2/h^3)$ (here $\Delta$ is the grid spacing and $h$ the bandwidth), so the FFT/binned criterion differs from the exact plug–in criterion by a uniform $o_p(1)$ term.


On the statistical side, we establish (i) \emph{consistency} of the FFT plug–in estimator $\hat\theta_n$ under (a) standard smoothness/undersmoothing for KDE, (b) a windowed body–tail control for the log‑likelihood, and (c) a local curvature (strong concavity) condition for the population objective around $\theta^*$; and (ii) a \emph{CLT} for smooth functionals after strengthening the discretization to $\sqrt{n}\,\varepsilon_n\!\to\!0$, where $\varepsilon_n\!=\!\tfrac{\|\mathbb K''\|_\infty}{4}\Delta_n^2/h_n^3$ is the binned‑FFT interpolation error. The proofs combine a uniform oracle LLN for the parametric part \parencite{vaart1998} with (a) uniform control of the FFT/KDE approximation on a growing body window and (b) crude but sufficient tail envelopes that shrink against $\log(1/h_n)$. The resulting consistency and CLT coincide with “what one would expect’’ if the background density were known, modulo the small plug–in remainder.

Classical robust statistics mitigates sensitivity by modifying the loss, bounding the influence function, or minimizing a robust divergence \parencite{Huber1981,Hampel1974,Maronna2019}. Our approach differs in these aspects: (1) we do not posit a global parametric model for all observations—only the dominant component is parametric, the rest is nonparametric; (2) the \emph{weights} arise from posterior membership in a mixture rather than from a downweighting scheme engineered via an influence curve. In this sense our method is closer to the semiparametric mixture literature, except here the unknown component is neither known nor shape‑constrained, and we leverage a fast nonparametric plug‑in evaluated inside the EM algorithm.

\section{Estimating a single dominant component}
We first consider the problem of estimating a single dominating subpopulation from a heterogeneous population. Towards this, we model the data $Y_1,\dots,Y_n$ through a two component mixture $f,$ which can be written as
\begin{equation}
\label{eq:model}
 f(y) \;=\; \pi_0 f_0(y;\theta) + (1-\pi_0) f_1(y),
\end{equation}
where $f_0(\cdot\,;\theta)$ is our main component of interest and it is modeled through a parametric family with unknown parameter $\theta\in\Theta\subset\R^q,$ and $f_1$ is kept unspecified, modeling the remaining mixture of the other background components. The mixing weights $\pi_0$ and $1-\pi_0$ are unknown but constant. We observe $n$ i.i.d. samples $\mathbf y = (y_1,\cdots, y_n)$ from this mixture.\\

\medskip
\noindent
The average log-likelihood function of \eqref{eq:model} is  
\begin{equation}
\mathcal L(\Theta ; \mathbf{y}) 
= \frac{1}{n} \sum_{i=1}^n 
\log \Bigl\{ 
\pi_0f_0(y_i; \theta) 
+ (1-\pi_0)f_1(y_i) 
\Bigr\},\qquad \Theta=(\pi_0,\theta)
\end{equation}
We cannot optimize the above expression directly because $f_1(y)$ is unknown. For this reason we approximate $f_1$ via a kernel density estimator 
\begin{equation}
\tilde f^{\mathrm{exact}}_{1,n}(y; \omega) 
= \frac{1}{h_n} \sum_{j=1}^n 
\omega_j \, \mathbb{K}\!\left( \frac{y_j - y}{h_n} \right),
\end{equation}
where $\mathbb K(\cdot)$ is a kernel function satisfying Condition \ref{cond:kernel}, $\omega_j$'s are weighting parameters such that $\sum_{i=1}^{n}\omega_j = 1$ and $h_n > 0$ is the bandwidth. $\omega_j$ is a function of the observed data $y_j.$ For notational simplicity we drop this explicit dependence in our expressions.\\

\medskip
\noindent
\begin{cond}[Kernel properties]\label{cond:kernel}
The kernel function $\mathbb{K}(z)$ satisfies
\[
\int \mathbb{K}(z)\,dz = 1, 
\quad 
\int \mathbb{K}(z) \, z \, dz = 0, 
\quad 
\int \mathbb{K}(z) \, z^2 \, dz < \infty.
\]
\end{cond}

\noindent
The \emph{working log-likelihood} using this plug-in estimate for the nonparametric density $\tilde f_1$ is
\[
 \tilde{ \mathcal{L}}_n(\Theta,\omega; \mathbf y) = \frac1n\sum_{i=1}^n \log\Big\{\pi_0 f_0(y_i;\theta) + (1-\pi_0)\tilde f^{\mathrm{exact}}_{1,n}(y_i)\Big\},\qquad \Theta=(\pi_0,\theta).
\]
Let $\Theta^*$ denote the truth of $\Theta$, and $f_1^*$ denote the truth of $f_1$, then from \cite{You2024sequential} we have the following theorem\\
\medskip
\noindent
\begin{theorem}
\label{thm:consistency}
If condition \ref{cond:kernel} holds and $h_n \to 0$, $n h_n \to \infty$ as $n \to \infty$, 
then under some mild regularity conditions, there exists $\hat{\omega}$ such that the maximizer $\hat{\Theta}$ of 
$\tilde{\mathcal{L}}(\Theta,\hat{\omega})$ converges to 
$\Theta^\ast$ in probability.
\end{theorem}
\noindent
Under standard regularity and boundedness conditions ~\cite{You2024sequential} guarantee the identifiability of $\hat{\Theta}$ and one can only identify the component that is large enough. The component with very small weight cannot be identified. A concavity assumption further guarantees the uniqueness of the parameter estimate.
\\~\\
The algorithm for the likelihood maximization proposed in \cite{You2024sequential} requires evaluating $\tilde f^{\mathrm{exact}}_{1,n}(y; \omega) 
= \frac{1}{h_n} \sum_{j=1}^n 
\omega_j \, \mathbb{K}\!\left( \frac{y_j - y}{h_n} \right)$ at all $Y_i$. This leads to a computational complexity of $\mathcal O(n^2)$ per iteration. In finite samples their algorithm can also suffer from role flips where the parametric block absorbs the tight minority while the KDE models the diffuse majority. To counter such issues, we propose the following ideas.
%
%
\subsection{Fast weighted KDE via binning and FFT}
The kernel density estimator of the unknown component $f_1$ is
\[
\tilde{f}_1(y) = \sum_{j=1}^n \omega_j\, \mathbb K_{h_n}(y - Y_j),
\qquad \mathbb K_{h_n}(u) = \frac{1}{h_n} \mathbb K\!\left(\frac{u}{h_n}\right).
\]
where $\sum_i \omega_i = 1.$ A direct evaluation of this weighted KDE at all \(n\) data points costs \(O(n^2)\) per EM iteration. We exploit that a binned KDE is a discrete convolution and can be computed in \(O(M \log M)\) time on a grid through Fast Fourier Transform (FFT).

\paragraph{Linear binning.} Fix an equally spaced grid $x_1<\cdots<x_M$ with spacing $\Delta$ such that $x_m = x_1 + (m-1)\Delta$, $m=1,\dots,M$ and an array $c[1:M]$ initialized to zeros.
For each $Y_i$, define $p_i = (Y_i - x_1)/\Delta$, $j_i = \lfloor p_i\rfloor$, $\alpha_i = p_i - j_i \in [0,1)$.  
We split the weights linearly as:
\[
c[j_i] \!+\!= \omega_i(1-\alpha_i),
\qquad
c[j_i+1] \!+\!= \omega_i \alpha_i.
\]
Equivalently,
\[
c[m]\ :=\ \sum_{i=1}^n \omega_i\Big[(1-\alpha_i)\mathbf 1\{m=j(i)\}+\alpha_i\mathbf 1\{m=j(i)+1\}\Big],
\qquad m=0,\dots,M-1.
\]
Then $c[m]\ge0$ and $\sum_{m=0}^{M-1}c[m]=1$.

\paragraph{Discrete convolution on the grid via FFT.}
At grid points $x_m,$ we have $x_m=x_0+m\Delta$ for $m=0,\ldots,M-1$. Define the sampled kernel
$k_h[r]:=\mathbb K_{h_n}(r\Delta)$ for integer lags $r\in\mathbb Z$.
The binned KDE on the grid is the \emph{linear convolution}
\[
\tilde f^{\mathrm{FFT}}_{1,n}(x_m)
=\sum_{\ell=0}^{M-1} c[\ell]\;k_h[m-\ell],
\qquad m=0,\ldots,M-1.
\]

To compute this with an FFT, choose a kernel
radius $R\in\mathbb N$,
set $L:=2R+1$, and pick an FFT length
$P\ \ge\ M+L-1
\quad\text{(in practice, }P=2^{\lceil \log_2(M+L-1)\rceil}\text{).}$
Construct two length‑$P$ arrays:
\[
\tilde c[p]=
\begin{cases}
c[p], & 0\le p\le M-1,\\
0, & M\le p\le P-1,
\end{cases}
\qquad
\tilde k[p]=
\begin{cases}
k_h[0], & p=0,\\
k_h[p], & 1\le p\le R,\\
k_h[-(P-p)], & P-R\le p\le P-1,\\
0, & \text{otherwise}.
\end{cases}
\]
The negative lags $k_h[-1],\ldots,k_h[-R]$ are stored at positions
$P\!-\!1,\ldots,P\!-\!R$. If $K$ is symmetric then $k_h[-r]=k_h[r]$.

Compute the $P$‑point FFTs and inverse FFT:
\[
C=\mathrm{FFT}_P(\tilde c),\qquad
K=\mathrm{FFT}_P(\tilde k),\qquad
Z=C\odot K,\qquad
s=\mathrm{IFFT}_P(Z),
\]
where $\odot$ denotes elementwise multiplication. With the above padding and
placement, the first $M$ entries of $s$ equal the desired \emph{linear}
convolution values,
\[
\ \tilde f^{\mathrm{FFT}}_{1,n}(x_m)=s[m],\qquad m=0,\ldots,M-1.\
\]

\paragraph{Off-grid evaluation.}  
For sample points $y_i$, linearly interpolate between the nearest grid values $\tilde{f}(x_m)$ values. Define
\[\tilde f^{\mathrm{FFT}}_{1,n}(y_i) := (1 - \alpha_i) \tilde f^{\mathrm{FFT}}_{1,n}(x_{j(i)}) + \alpha_i \tilde f^{\mathrm{FFT}}_{1,n}(x_{j(i)+1}).\]

\paragraph{Per–iteration cost.}
One EM iteration has three steps: (i) linear binning of $n$ weighted points into $M$ grid bins ($O(n)$); (ii) one zero–padded FFT to compute the grid convolution ($O(P\log P)$); (iii) linear interpolation back to the $n$ data points ($O(n)$). The FFT length only needs to satisfy $P\ge M+L-1$, where $L$ is the number of relevant kernel taps on the grid. Because the kernel is localized, $L\ll M$, so $P=M+O(L)=\Theta(M)$ and the FFT stage is $O(M\log M)$. Hence the per–iteration time is
\[
T_{\mathrm{iter}}=O(n)+O(M\log M)+O(n)=O\!\big(n+M\log M\big),
\]
with $O(M)$ memory. The grid size scales as $M\approx \mathrm{range}/\Delta_n$, so $M$ grows with $n$ only through the chosen grid step $\Delta_n$. In practical choices $M=o(n)$, making the overall iteration effectively linear in $n$.

\paragraph{Approximation error bound.}In this FFT implementation due to first \emph{binning} the data onto a grid, convolving these binned counts with the rescaled kernel $\mathbb K_{h_n}$ on the grid via FFT, and recovering $\tilde f^{\mathrm{FFT}}_{1,n}(y)$ for arbitrary $y$ by linear interpolation, we encounter approximation error in two stages. A second–order Taylor expansion shows that both the binning step and the interpolation step introduce at most a $O(\Delta_n^2/h_n^3)$ error per kernel term. Summing over the weights and using $\sum_i\hat\omega_i=1$ leads to a uniform bound on the discrepancy between the exact weighted KDE $\tilde f^{\mathrm{exact}}_{1,n}$ and its FFT approximation $\tilde f^{\mathrm{FFT}}_{1,n}$ over both the sample points $Y_i$ and the whole grid interval.
\medskip
\begin{lemma}
\label{lem:fft-binning-interp-0based}
Assume $\mathbb K\in C^2(\R)$ be bounded with $\|\mathbb K''\|_\infty<\infty$ and $\K_{h_n}(u):=h_n^{-1}\K(u/h_n)$. Then
\[
\max_{1\le i\le n}\big|\tilde f^{\mathrm{FFT}}_{1,n}(Y_i)-\tilde f^{\mathrm{exact}}_{1,n}(Y_i)\big|
\ \le\ \frac{1}{4}\,\|\mathbb K''\|_\infty\,\frac{\Delta_n^2}{h_n^3}.
\]
Moreover, from definition of $\tilde f^{\mathrm{FFT}}_{1,n}(y)$ through linear interpolation between neighboring grid values, we have
\[
\sup_{y\in[x_0,x_{M-1}]}\big|\tilde f^{\mathrm{FFT}}_{1,n}(y)-\tilde f^{\mathrm{exact}}_{1,n}(y)\big|
\ \le\ \frac{1}{4}\,\|\mathbb K''\|_\infty\,\frac{\Delta_n^2}{h_n^3}.
\]
\end{lemma}
\subsection{Consistency and large sample properties}
\label{subsec:consistency-fft-nocf}
As stated before, we observe \emph{i.i.d.}\ data $Y_1,\dots,Y_n$ from the one–dimensional mixture
\[
f(y)\;=\;\pi_0 f_0(y;\theta)+\pi_1 f_1(y),\qquad
\pi_0\in(0,1),\ \pi_1:=1-\pi_0,
\]
where $\theta^*\in\mathcal A\subset\R$ and we maximize over a compact $\mathcal C \subset \mathcal A$ containing $\theta^*$.

\begin{theorem}
\label{thm:consistency-fft-nocf}
Under Assumptions~\ref{ass:model-id}, \ref{ass:K-h-delta}, \ref{ass:kde-sup}, \ref{ass:window-nocf},
the location estimate obtained from maximizing the FFT plug-in log likelihood
\[
\hat\theta_n\in\arg\max_{\theta\in\mathcal C} \tilde L_n(\theta).
\]
converges to the true location parameter $\theta^*$ in probability, i.e.,
\[
\hat\theta_n\ \xrightarrow{p}\ \theta^*.
\]
\end{theorem}


\begin{theorem}
\label{thm:clt-fft-nocf}
Assume the conditions of Theorem~\ref{thm:consistency-fft-nocf} together with
Assumptions~\ref{ass:diff-oracle-nocf},
\ref{ass:plugin-tail-nocf},
\ref{ass:targeted-nocf} and \ref{ass:remainder-sep-nocf}.
Then
\[
\sqrt{n}\,(\hat\theta_n-\theta^*)\ \rightsquigarrow\ \mathcal N\!\big(0,\,I(\theta^*)^{-1}\big).
\]
\end{theorem}

\subsection{Estimation algorithm}
In order to obtain the FFT plug–in estimator $\hat\theta$ and the associated weights $\hat\omega=(\hat\omega_1,\ldots,\hat\omega_n)^\top$ for the nonparametric component, we use the EM algorithm on this approximate FFT plug-in log-likelihood.
\begin{proposition}
\label{prop:fft-em-iterates-final}
At iteration $t$, given $(\theta^{(t)},\pi_0^{(t)})$ and a current background approximation
$\tilde f^{\mathrm{FFT},(t)}_{1,n}$, define $\pi_1^{(t)}:=1-\pi_0^{(t)}$ and perform:

\medskip\noindent
\textbf{E--step (posterior probabilities).}
\[
r_i^{(t)}\ :=\
\frac{\pi_0^{(t)}\,f_0(Y_i;\theta^{(t)})}
     {\pi_0^{(t)}\,f_0(Y_i;\theta^{(t)})+\pi_1^{(t)}\,\tilde f^{\mathrm{FFT},(t)}_{1,n}(Y_i)},
\qquad
s_i^{(t)}:=1-r_i^{(t)},
\]
where the off--grid evaluation is
\[
\tilde f^{\mathrm{FFT},(t)}_{1,n}(Y_i)
=(1-\alpha_i)\,\tilde f^{\mathrm{FFT},(t)}_{1,n}(x_{j(i)})
+\alpha_i\,\tilde f^{\mathrm{FFT},(t)}_{1,n}(x_{j(i)+1}).
\]

\medskip\noindent
\textbf{M--step for $\pi_0$ and $\theta$.}
Let
\[
Q^{(t)}(\pi_0,\theta)\ :=\ \frac1n\sum_{i=1}^n
\Big\{ r_i^{(t)}\,[\log \pi_0+\log f_0(Y_i;\theta)]
     + s_i^{(t)}\,\log(1-\pi_0)\Big\}.
\]
Then
\[
\pi_0^{(t+1)}=\frac{1}{n}\sum_{i=1}^n r_i^{(t)},\qquad
\theta^{(t+1)}\in\arg\max_{\theta}\ \frac1n\sum_{i=1}^n r_i^{(t)}\,\log f_0(Y_i;\theta).
\]
Equivalently, $\theta^{(t+1)}$ solves the weighted score equation
\[
\frac1n\sum_{i=1}^n r_i^{(t)}\,\partial_\theta\log f_0(Y_i;\theta)=0.
\]
\medskip\noindent
\textbf{Background update (for the next E--step).}
Normalize the posterior probabilities
\[
\omega_i^{(t)}:=\frac{s_i^{(t)}}{\sum_{j=1}^n s_j^{(t)}},\qquad \omega_i^{(t)}\ge0,\ \sum_i\omega_i^{(t)}=1,
\]
bin them to the grid
\[
c^{(t)}[m]\ :=\ \sum_{i=1}^n \omega_i^{(t)}
\Big[(1-\alpha_i)\mathbf 1\{m=j(i)\}+\alpha_i\mathbf 1\{m=j(i)+1\}\Big],
\quad m=0,\ldots,M-1,\quad \sum_m c^{(t)}[m]=1,
\]
and compute the grid convolution by zero--padded FFT (linear convolution)
\[
\tilde f^{\mathrm{FFT},(t+1)}_{1,n}(x_m)=\sum_{\ell=0}^{M-1} c^{(t)}[\ell]\;\mathbb K_{h_n}(x_m-x_\ell),\qquad m=0,\ldots,M-1.
\]
For off--grid points $Y_i$ in the next E--step use the linear interpolation
\[
\tilde f^{\mathrm{FFT},(t+1)}_{1,n}(Y_i)
=(1-\alpha_i)\,\tilde f^{\mathrm{FFT},(t+1)}_{1,n}(x_{j(i)})
+\alpha_i\,\tilde f^{\mathrm{FFT},(t+1)}_{1,n}(x_{j(i)+1}).
\]
\end{proposition}

Our FFT-accelerated approximate EM algorithm for locating the dominant subpopulation location is presented below.

\begin{algorithm}[ht]
\caption{FFT accelerated EM for a two–component semiparametric mixture}
\label{alg:fft-plugin-em-compact}
\begin{algorithmic}
\Require Data $\{Y_i\}_{i=1}^n$; kernel $\mathbb K$; bandwidth $h_n$; FFT grid $\{x_j\}_{j=0}^{M-1}$; tolerance $\texttt{tol}$; maximum iterations $\texttt{maxit}$.
Randomly split the data by assigning $r_i^{(0)} \sim Ber(0.5).$ Initialize $\pi_0^{(0)}, \, \theta^{(0)}$ from the \textbf{M-step} in Proposition~\ref{prop:fft-em-iterates-final}
\Repeat
  \State \textbf{Background update:} Normalise weights
         $\omega_i^{(t)} \gets s_i^{(t)}\big/\sum_{k=1}^n s_k^{(t)}$;
         bin onto $\{x_j\}$ and convolve with $\mathbb K_{h_n}$ via FFT to get 
         $\{\tilde f_{1}^{(t)}(x_j)\}_{j=0}^{M-1}$
  \State \textbf{Off–grid interpolation:} For each $i$, find $j(i)$ with
         $Y_i\in[x_{j(i)},x_{j(i)+1})$. $\alpha_i = (Y_i-x_{j(i)})/\Delta$ and set\\
         $\tilde f_{1}^{(t)}(Y_i) \gets (1-\alpha_i)\tilde f_{1}^{(t)}(x_{j(i)})+\alpha_i\tilde f_{1}^{(t)}(x_{j(i)+1})$
  \State \textbf{E–step (posterior probs):}
         \[
           r_i^{(t)} \gets 
           \frac{\pi_0^{(t)} f_0(Y_i;\theta^{(t)})}{\pi_0^{(t)} f_0(Y_i;\theta^{(t)})+(1-\pi_0^{(t)})\,\tilde f_{1}^{(t)}(Y_i)},
           \qquad s_i^{(t)} \gets 1-r_i^{(t)} \quad (i=1,\dots,n)
         \]
  \State \textbf{M–step:}
         \[
           \pi_0^{(t+1)} \gets \frac{1}{n}\sum_{i=1}^n r_i^{(t)},
           \qquad
           \theta^{(t+1)} \in \arg\max_{\theta\in\mathcal C}
           \sum_{i=1}^n r_i^{(t)}\log f_0(Y_i;\theta)
         \]
  \State \textbf{Working log–likelihood:}
         \[
           \tilde{\mathcal L}^{(t+1)} \gets \frac{1}{n}\sum_{i=1}^n 
           \log\!\big\{\pi_0^{(t+1)} f_0(Y_i;\theta^{(t+1)}) + (1-\pi_0^{(t+1)})\,\tilde f_{1}^{(t)}(Y_i)\big\}
         \]
  \State $t \gets t+1$
\Until{$|L^{(t)}-L^{(t-1)}| < \texttt{tol}$ \textbf{ or } $t=\texttt{maxit}$}
\State \textbf{Output:} $\hat\theta \gets \theta^{(t)}$, \; $\hat\pi_0 \gets \pi_0^{(t)}$.
\end{algorithmic}
\end{algorithm}

\paragraph{Notes on practical usage.}
To prevent the nonparametric component from latching onto a diffuse majority or a sharp minority, we recommend the following heuristic. First, apply a simple one–dimensional $k$–means clustering to the data and split it into ``left’’ and ``right’’ clusters. Then initialize the posterior probabilities $r_i^{(0)}$ twice: once using the proportion in the smaller cluster and once using the proportion in the larger cluster. Run Algorithm 1 under both initializations and retain the estimate $\hat\theta$ corresponding to the fit with the larger estimated mixing weight $\hat\pi_0$ on the parametric component.

\section{Estimating more than 2 components}
\label{subsec:sequential-partitioning}

In a practical, more general setting, there could be multiple significant subpopulations of interest in a heterogeneous population. Motivated from ~\cite{You2024sequential}, we consider a semiparametric mixture in which some (unknown) number of 
\emph{parametric} peaks from a known family $f_0(\cdot;\theta)$ are embedded in an 
otherwise unrestricted background. The goal is to recover the parametric peaks 
sequentially, one stage at a time, using a two–component FFT plug–in EM at each stage, 
and to provide guarantees that the stage–$\ell$ procedure consistently estimates 
the $\ell$–th true peak locations.

We consider our data $Y_1,\dots,Y_n$ coming from a density $f$ on $\R$ of the form
\[
f(y)
\;=\;\sum_{j=1}^{J} \alpha_j f_0(y;\theta_j)
   \;+\;
   \Big(1-\sum_{j=1}^{J}\alpha_j\Big)\,f_{\mathrm{bg}}(y),
\qquad y\in\R,
\]
with $\alpha_j\in(0,1)$, $\sum_{j=1}^{J}\alpha_j<1$, and
$\theta_j\in\mathcal A\subset\R$. Assume that for some
$\ell\ge2$ we have already obtained consistent estimators
$(\hat\theta_j,\hat\alpha_j)$ of $(\theta_j,\alpha_j)$ for $j<\ell$.


\paragraph{Residual mixture at stage $\ell$.}
Define the stage--$\ell$ residual density
\[
f^{(\ell)}(y)
:=\frac{f(y)-\sum_{j<\ell}\alpha_j f_0(y;\theta_j)}
        {1-\sum_{j<\ell}\alpha_j},\qquad y\in\R.
\]
Then $f^{(\ell)}$ is a probability density. $f^{(\ell)}$
admits a further two--component decomposition
\begin{equation}\label{eq:residual-two-comp}
f^{(\ell)}(y)
=\alpha_\ell^{(\ell)} f_0(y;\theta_\ell)
 +(1-\alpha_\ell^{(\ell)}) f_{\mathrm{bg}}^{(\ell)}(y),
\end{equation}
where
\[
\alpha_\ell^{(\ell)}:=\frac{\alpha_\ell}{1-\sum_{j<\ell}\alpha_j}\in(0,1),
\]
and $f_{\mathrm{bg}}^{(\ell)}$ is a (completely unspecified) background density
at stage~$\ell$. Our target parameter at stage~$\ell$ is
\[
\eta_\ell:=(\theta_\ell,\alpha_\ell^{(\ell)})\in\mathcal C\times\Pi_{\min},
\]
where $\mathcal C\subset\mathcal A$ is compact with
$\theta_\ell\in\mathcal C$ and
$\Pi_{\min}:=[\underline\alpha_\ell,1-\underline\alpha_\ell]$ for some
$0<\underline\alpha_\ell<\alpha_\ell^{(\ell)}<1$.


\paragraph{Estimated residual weights and plug--in criterion.}
Let $(\hat\theta_j,\hat\alpha_j)$, $j<\ell$, denote the stage--$(1,\dots,\ell-1)$
estimators, and let $\hat f$ be an estimator of $f$ (for instance, the fitted
mixture density from earlier stages). Define the unnormalised residual weights
\[
\widehat v_i^{(\ell)}
:=1-\sum_{j<\ell}\hat\alpha_j\,
  \frac{f_0(Y_i;\hat\theta_j)}{\hat f(Y_i)},
\qquad i=1,\dots,n,
\]
and assume $\widehat v_i^{(\ell)}\ge0$ a.s.\ for all $i$ and large $n$. Normalize
to obtain probability weights
\[
\widehat w_i^{(\ell)}
:=\frac{\widehat v_i^{(\ell)}}{\sum_{k=1}^n \widehat v_k^{(\ell)}},
\qquad i=1,\dots,n.
\]
These weights approximate sampling from $f^{(\ell)}$:
they downweight observations explained by earlier peaks.

Let $\tilde f_{\ell,n}^{\mathrm{FFT}}$ be the stage--$\ell$ background density
estimator, constructed by a weighted KDE using $\{\widehat w_i^{(\ell)}\}$ as in
Section~\ref{subsec:consistency-fft-nocf}, accelerated by FFT and linear binning.
Define, for $\eta=(\theta,\alpha)\in\mathcal C\times\Pi_{\min}$,
\[
\ell^{(\ell)}(y;\eta,u)
:=\log\{\alpha f_0(y;\theta)+(1-\alpha)u(y)\},
\]
and the stage--$\ell$ plug--in criterion
\[
\tilde {\mathcal{L}}_n^{(\ell)}(\eta)
:=\sum_{i=1}^n \widehat w_i^{(\ell)}
  \,\ell^{(\ell)}\big(Y_i;\eta,\tilde f_{\ell,n}^{\mathrm{FFT}}\big).
\]
The stage--$\ell$ estimator is any maximiser
\[
\hat\eta_{\ell,n}\in\arg\max_{\eta\in\mathcal C\times\Pi_{\min}}
\tilde L_n^{(\ell)}(\eta).
\]

The corresponding oracle criterion replaces $\tilde f_{\ell,n}^{\mathrm{FFT}}$
by $f_{\mathrm{bg}}^{(\ell)}$ and $P_n^{(\ell)}:=\sum_i\widehat w_i^{(\ell)}\delta_{Y_i}$
by $f^{(\ell)}$:
\[
\mathcal L^{(\ell)}(\eta)
:=\E_{f^{(\ell)}}\big[\ell^{(\ell)}(Y;\eta,f_{\mathrm{bg}}^{(\ell)})\big],
\qquad \eta\in\mathcal C\times\Pi_{\min}.
\]

As before, denote $\mathbf{\theta}^* = (\theta_1^*, \cdots, \theta_J^*) $ as the truth of $\mathbf{\theta}$. Then the parameter estimates $\hat\theta_1, \cdots, \hat \theta_J$ enjoy the following statistical properties.

\begin{theorem}[Stage--$\ell$ consistency]\label{thm:consistency-stage}
Under Assumptions~\ref{ass:model-id-stage}--\ref{ass:log-env-plugin-stage},
the stage--$\ell$ FFT plug--in MLE is consistent:
\[
\hat\eta_{\ell,n}=(\hat\theta_{\ell,n},\hat\alpha_{\ell,n})
\ \xrightarrow{p}\ \eta_\ell^*=(\theta_\ell^*,\alpha_\ell^{(\ell)}).
\]
\end{theorem}



\begin{theorem}[Stage--$\ell$ CLT]
\label{thm:clt-stage}
Suppose the conditions of Theorem~\ref{thm:consistency-stage} hold, together with
Assumptions~\ref{ass:diff-info-stage}--\ref{ass:plugin-oracle-stage}.
Then
\[
\sqrt{n}\big(\hat\eta_{\ell,n}-\eta_\ell^*\big)
\ \rightsquigarrow\ \mathcal N\big(0,I^{(\ell)}(\eta_\ell^*)^{-1}\big).
\]
In particular, the marginal distribution of
$\sqrt{n}(\hat\theta_{\ell,n}-\theta_\ell^*)$ is normal with covariance given by
the $(\theta,\theta)$--block of $I^{(\ell)}(\eta_\ell^*)^{-1}$.
\end{theorem}
\newpage
\begin{algorithm}[H]
\caption{Stage--$\ell$ FFT accelerated EM (sequential peak extraction)}
\label{alg:seq-fft-plugin-em}
\begin{algorithmic}[1]
\Require Data $\{Y_i\}_{i=1}^n$; stage $\ell\ge1$; previous components $(\hat\alpha_j,\hat\theta_j)_{j<\ell}$ (empty if $\ell=1$); kernel $\mathbb K$ with bandwidth $h_n$; FFT grid $\{x_j\}_{j=0}^{M-1}$ (spacing $\Delta$); tolerance \texttt{tol}; maximum iterations \texttt{maxit}.
\Ensure Stage--$\ell$ estimates $(\hat\alpha_\ell,\hat\theta_\ell)$ and background FFT KDE $\{\tilde f_\ell^{\mathrm{FFT}}(x_j)\}_{j=0}^{M-1}$.

\State \textbf{Residual weights}
\State $\hat f_{1:(\ell-1)}(y)\gets \sum_{j<\ell}\hat\alpha_j f_0(y;\hat\theta_j)$, 
       $\hat\rho_{\ell-1}\gets 1-\sum_{j<\ell}\hat\alpha_j$.
\State Raw residual scores
       $\hat v_i^{(\ell)}\gets \max\{0,\,1-\hat f_{1:(\ell-1)}(Y_i)/\max(\hat f_{1:(\ell-1)}(Y_i),\varepsilon)\}$.
\State Normalised residual weights
       $\hat w_i^{(\ell)}\gets\hat v_i^{(\ell)}/\sum_k\hat v_k^{(\ell)}$.

\Statex
\State \textbf{Initialization}
\State Choose $\theta_\ell^{(0)}$ and
       $\alpha_\ell^{(\ell,0)}\in(0,1)$.
\State Set $r_i^{(\ell,0)}\gets s_i^{(\ell,0)}\gets \tfrac12$, $t\gets0$.

\Repeat
  \State \textbf{Background FFT KDE (stage $\ell$)}
  \State $\omega_i^{(\ell,t)}\gets
         \hat w_i^{(\ell)} s_i^{(\ell,t)}\Big/\sum_k \hat w_k^{(\ell)} s_k^{(\ell,t)}$.
  \State Bin $\{(Y_i,\omega_i^{(\ell,t)})\}$ on $\{x_j\}$ and convolve with $\mathbb K_{h_n}$ via FFT
         to get $\{\tilde f_\ell^{(\mathrm{FFT},t)}(x_j)\}$.
  \State Interpolate linearly to obtain $\tilde f_\ell^{(\mathrm{FFT},t)}(Y_i)$ for all $i$.

  \Statex
  \State \textbf{E--step (stage $\ell$)}
  \For{$i=1,\dots,n$}
    \State $f_i^{(\ell,t)}\gets
        \alpha_\ell^{(\ell,t)} f_0(Y_i;\theta_\ell^{(t)})
       +(1-\alpha_\ell^{(\ell,t)})\tilde f_\ell^{(\mathrm{FFT},t)}(Y_i)$.
    \State $r_i^{(\ell,t+1)}\gets
       \alpha_\ell^{(\ell,t)} f_0(Y_i;\theta_\ell^{(t)})/f_i^{(\ell,t)}$,
       \quad $s_i^{(\ell,t+1)}\gets 1-r_i^{(\ell,t+1)}$.
  \EndFor

  \Statex
  \State \textbf{M--step (stage $\ell$)}
  \State $\alpha_\ell^{(\ell,t+1)}\gets\sum_{i=1}^n \hat w_i^{(\ell)} r_i^{(\ell,t+1)}$.
  \State $\theta_\ell^{(t+1)}\in\arg\max_{\theta\in\overline{\mathcal A}}
         \sum_{i=1}^n \hat w_i^{(\ell)} r_i^{(\ell,t+1)}\log f_0(Y_i;\theta)$.

  \Statex
  \State \textbf{Working plug--in log--likelihood}
  \State $\tilde{\mathcal L}_\ell^{(t+1)}\gets
         \sum_{i=1}^n \hat w_i^{(\ell)}
         \log\!\big\{
           \alpha_\ell^{(\ell,t+1)} f_0(Y_i;\theta_\ell^{(t+1)})
          +(1-\alpha_\ell^{(\ell,t+1)})\tilde f_\ell^{(\mathrm{FFT},t)}(Y_i)
         \big\}$.
  \State $t\gets t+1$.
\Until{$|\tilde{\mathcal L}_\ell^{(t)}-\tilde{\mathcal L}_\ell^{(t-1)}|<\texttt{tol}$ or $t=\texttt{maxit}$}

\State \textbf{Output (stage $\ell$)}
\State $\hat\theta_\ell\gets\theta_\ell^{(t)}$, \quad
       $\hat\alpha_\ell^{(\ell)}\gets\alpha_\ell^{(\ell,t)}$, \quad
       $\hat\alpha_\ell\gets\hat\alpha_\ell^{(\ell)}\big(1-\sum_{j<\ell}\hat\alpha_j\big)$,
       \quad $\tilde f_\ell^{\mathrm{FFT}}(x_j)\gets\tilde f_\ell^{(\mathrm{FFT},t)}(x_j)$.
\end{algorithmic}
\end{algorithm}

\section{Conclusion}
This work set out as a robust location estimation problem through likelihood–based mixture estimation, in a different spirit from classical techniques in robust statistics. We formalized a semiparametric contamination model, developed an FFT–accelerated plug–in EM algorithm that estimates the parametric peak while computing the unknown nonparametric background by a fast kernel density estimator. We have also proved consistency and asymptotic normality of the resulting estimates, both in a single–stage setting and in a sequential “peel–off’’ scheme that extracts multiple dominant modes one by one. Our estimates provide similar statistical large sample guarantees as the usual maximum likelihood estimates. Compared to vanilla EM, our proposed method leverages FFT–plug-in KDE to achieve near–linear or $(G\log G)$ complexity per iteration, making it practical for large-scale problems. Together, these features provide a principled and interpretable way to decompose complex and heterogeneous multi–modal distributions into a small number of well–identified peaks plus an unrestricted remainder, without sacrificing the familiar EM workflow that practitioners find easy to implement and extend.


\section{Experiments}
We now illustrate the empirical performance of our method on a range of synthetic and real-world datasets. The synthetic examples aim to probe the ability of the algorithm to recover the dominant peak location under varying levels of contamination and multimodality, while keeping full control over the ground truth. We then turn to real datasets where the experimenter has no control over the underlying structure. We use these applications to assess the robustness of the FFT plug-in EM algorithm. These experiments provide a concrete picture of how the proposed approach behaves in finite samples and in practical settings. First we focus on the estimation of a single dominant component.

\subsection{Simulated datasets}
We first consider a mixture with a dominant component consisting of $60\%$ $N(10,1)$ and a contaminating component of $40\%$ $N(15,1)$. We generate $n = 500$ observations from this model and run Algorithm 1 on the resulting sample. The estimated centre of the dominant population is $\hat\mu = 9.98$, which is extremely close to the true centre $10$.
\begin{figure}[H]
  \centering
  \includegraphics[width=0.7\linewidth]{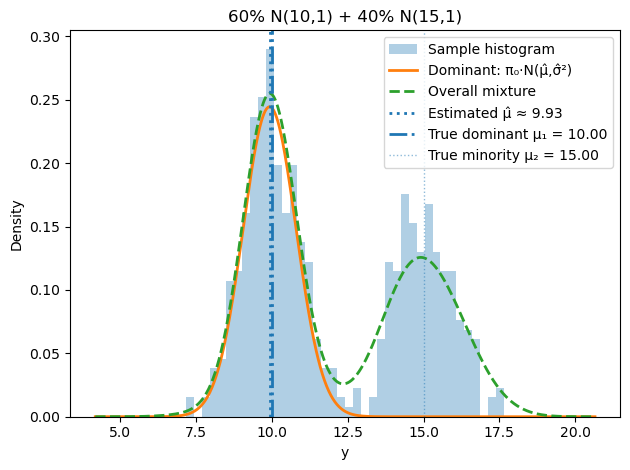} 
  \caption{60-40 mixture of N(10,1) and N(15,1)}
  \label{exp1}
\end{figure}
Now we consider a mixture with a dominant component consisting of $60\%$ $N(10,16)$ -- a diffuse and ``spread-out" population and a relatively sharp contaminating component of $40\%$ $N(15,1)$. We sample $n = 1000$ observations from this model and run Algorithm 1 on the resulting sample. The estimated centre of the dominant population is $\hat\mu = 9.28$, which is remarkably close to the true centre $10$, despite the high variance of the primary population.
\begin{figure}[H]
  \centering
  \includegraphics[width=0.7\linewidth]{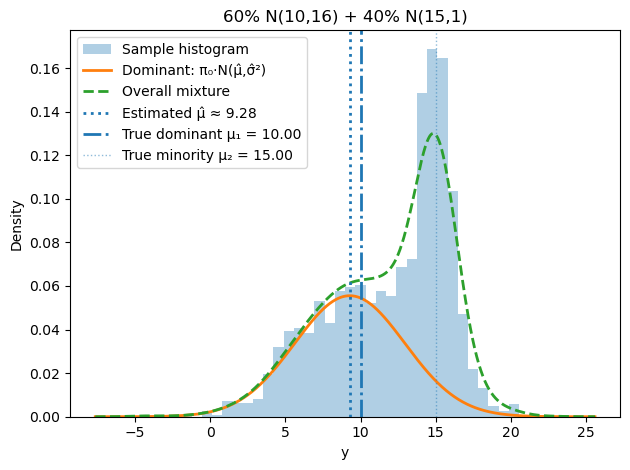} 
  \caption{60-40 mixture of N(10,16) and N(15,1)}
  \label{exp2}
\end{figure}
Let us demonstrate our practical heuristic strategy now. Consider a mixture with a dominant component consisting of $60\%$ $N(10,1)$ and a diffuse contaminating component of $40\%$ $N(20,16)$. We again generate $n = 1000$ observations from this model and run Algorithm 1 on the resulting sample. The nonparametric part in the vanilla weighted EM algorithm \parencite{You2024sequential} absorbed the diffuse minority and estimated the centre of the dominant population as $\hat\theta = 19.91$, which is actually close to the centre of the contaminating component. While our practical method ran two simultaneous EM algorithms and returned $\hat\theta = 9.98$, which is indeed very close to the true location $10.$

\begin{figure}[H]
  \centering
  \begin{minipage}[t]{0.48\linewidth}
    \centering
    \includegraphics[width=\linewidth]{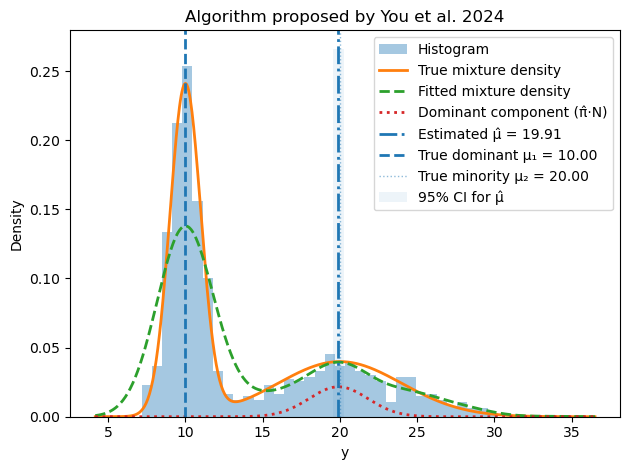}
    \caption{Flipped estimates in a 60-40 mixture of N(10,1) and N(20,16)}
    \label{exp4}
  \end{minipage}\hfill
  \begin{minipage}[t]{0.48\linewidth}
    \centering
    \includegraphics[width=\linewidth]{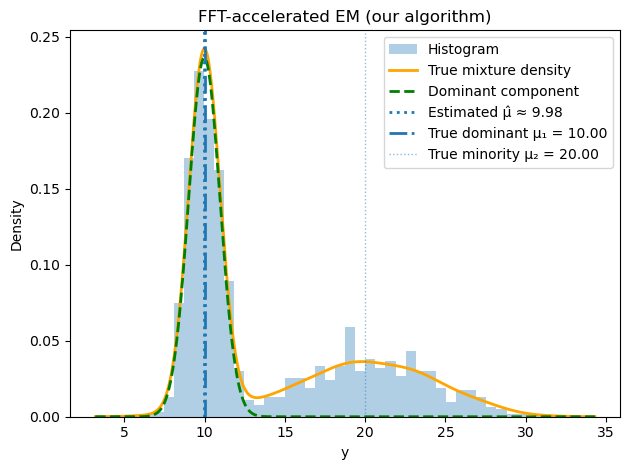}
    \caption{Correct estimates in a 60-40 mixture of N(10,1) and N(20,16)}
    \label{exp3,4}
  \end{minipage}
\end{figure}

A runtime performance for this setup is summarized in Table~\ref{tab:runtime}. All simulations were run on a standard consumer-grade computer. Note that in all the runs, the vanilla EM was unable to identify the true dominant component.
\begin{table}[ht]
    \centering
    \caption{Runtime comparison (seconds) of FFT accelerated EM vs. vanilla EM.}
    \label{tab:runtime}
    \vspace{1.5mm}
    \begin{tabular}{rcc}
        \hline
        Sample size & FFT accelerated EM (s) & Vanilla EM (s) \\
        \hline
        $n=500$   & 0.049159   & 1.391762   \\
        $n=1000$  & 0.043346   & 6.263643   \\
        $n=2000$  & 0.059048   & 18.146281  \\
        $n=3000$  & 0.059223   & 38.745702  \\
        $n=5000$  & 0.113147   & 163.754303 \\
        $n=7000$  & 0.132501   & 253.051952 \\
        $n=10000$ & 0.342913   & 1145.869190 \\
        $n=12000$ & 0.220001   & 2148.406933 \\
        $n=15000$ & 0.315351   & 6035.796481 \\
        \hline
    \end{tabular}
\end{table}

Consider a mixture with a dominant component consisting of $60\%$ \emph{Cauchy} with centre $0$, scale $2$ and a contaminating component of $40\%$ \emph{Cauchy} with centre $10$, scale $1$. We generate $n = 500$ samples from this model. But now, we model the parametric part through a normal family in our working log-likelihood and run Algorithm 1 on the sample. The estimated centre of the dominant population is $\hat\mu = -0.25$, which is very close to the true centre $0$, even after only approximately modeling the dominant component through a Normal parametric family.
\begin{figure}[H]
  \centering
  \includegraphics[width=0.7\linewidth]{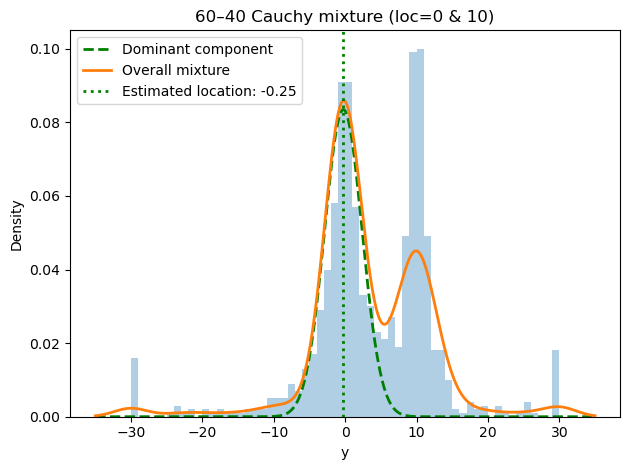}
  \caption{Cauchy mixture (60–40 at centers 0 \& 10, scales 2 \& 1)}
  \label{exp9}
\end{figure}
Again, we consider a mixture with a dominant component consisting of $60\%$ \emph{Cauchy} with centre $0$, scale $5$ (diffuse majority) and a contaminating component of $40\%$ \emph{Cauchy} with centre $10$, scale $1$ (tight minority). Generate $n = 1000$ observations from this model. Once again, we model the dominant component through a normal family in the working log-likelihood and run Algorithm 1 on the sample. The estimated centre of the dominant population is $\hat\mu = 1.08$, which is remarkably close to the true centre $0$, even when we only approximately model the larger component through a Normal family.
\begin{figure}[H]
  \centering
  \includegraphics[width=0.6\linewidth]{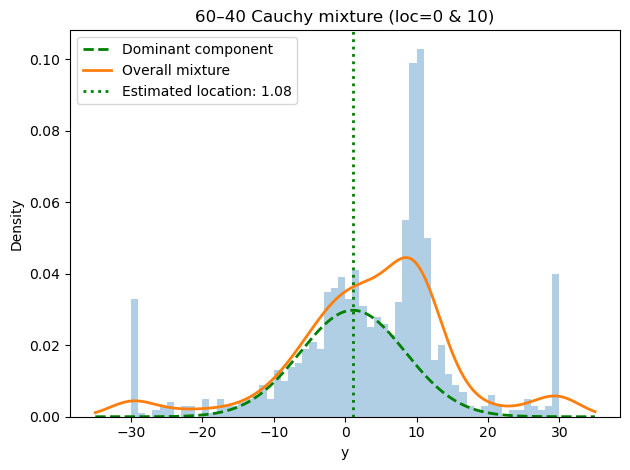} 
  \caption{Cauchy mixture (60–40 at centers 0 \& 10, scales 5 \& 1)}
  \label{exp10}
\end{figure}

\subsection{Real datasets}
This dataset records the lengths of Snapper fish, originally collected by \cite{cassie1954some}.
The purpose of the data is to distinguish among the possible age groups to which an individual fish might belong—specifically, whether it comes from the current year's spawning cohort or from that of the previous year.
The empirical distribution is visibly bimodal, with modes located near 5 and 7.75 inches.  
The histogram shows a more pronounced peak around 5 inches, suggesting that fish from the current year's spawning are represented in greater proportion than those from the earlier cohort.

Our algorithm estimates the average length of the dominant population to be $5.385$ inches, which is very close to estimates calculated using more sophisticated techniques in other robust statistics literature.
\begin{figure}[H]
  \centering
  \includegraphics[width=0.7\linewidth]{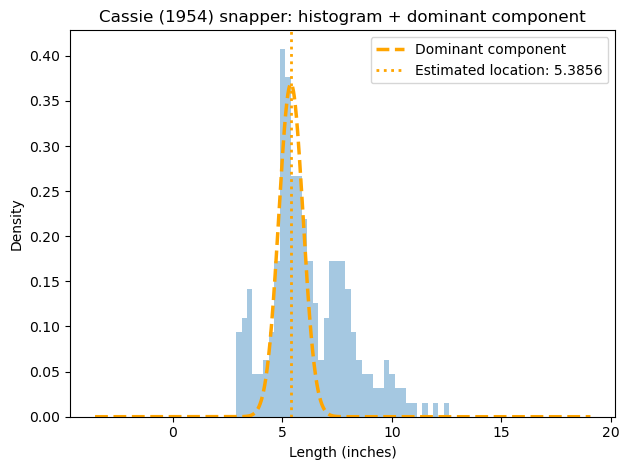} 
  \caption{Snapper data. Estimated mean 5.3856 inches}
  \label{exp5}
\end{figure}
The second dataset is Newcomb’s classic light–speed measurements, previously analyzed by \cite{Basu1998}, \cite{brown1993approximate}, and others.  
The data records the time (in nanoseconds) required for light to traverse 7442 meters at sea level, subtracted from a baseline value of 24800 nanoseconds.
The resulting histogram shows a prominent peak around 30, along with two well-known outliers at $-2$ and $-44$. We further shift the dataset by adding 44 to each observation so that all values become non-negative.
Clearly, the estimate from our algorithm is substantially close to the peak location of the larger population and coincides with other estimates from more sophisticated robust techniques.
\begin{figure}[H]
  \centering
  \includegraphics[width=0.7\linewidth]{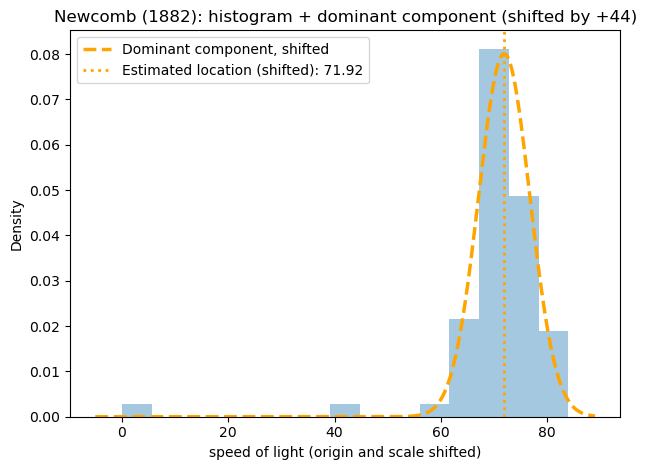} 
  \caption{Newcomb light speed data. Estimated mean 71.922}
  \label{exp6}
\end{figure}
The next dataset contains the width–to–length ratios of $20$ beaded rectangles crafted by the Shoshoni Indians and used as decorative elements on clothing and blankets \parencite{larsen2005introduction}.  
Our interest lies in estimating the central tendency of these ratios.  
In particular, one may ask whether the Shoshoni exhibited a preference for \emph{golden rectangles}; see, for example, \cite{Livio2002GoldenRatio} for a discussion of golden rectangles.  
If such a preference existed, then the width–to–length ratios of their rectangles would be expected to cluster around the golden ratio, approximately \(0.618\).

Our algorithm returns the centre of the dominant population as $0.63,$ quite close to $0.618$, suggesting empirical evidence in favour of the suspected preference of the Shoshoni people towards golden rectangles.
\begin{figure}[H]
  \centering
  \includegraphics[width=0.7\linewidth]{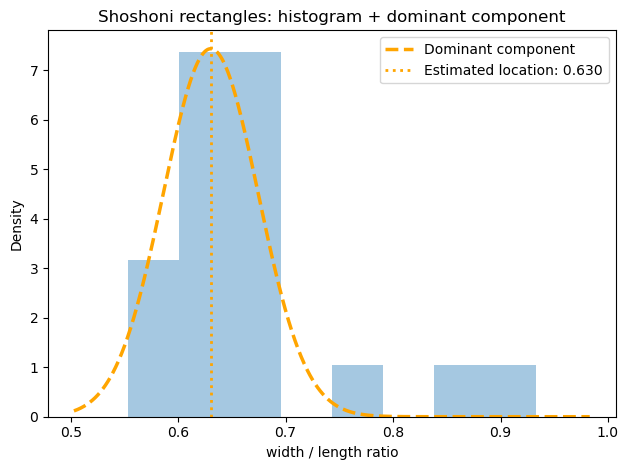} 
  \caption{Ratio data. Estimated mean 0.630293}
  \label{exp7}
\end{figure}
The fourth dataset consists of rainfall measurements (in millimeters) recorded every fourth day during the winter months (June, July, and August) of 1981–83 in Melbourne, Australia \cite{staudte2011robust}.  
The dataset includes a substantial outlier of approximately 300 mm.  
For our analysis, we use a scale-transformed version of these observations.

The estimate from our algorithm is reasonably close to the peak location of the larger population, despite the presence of extreme outliers.
\begin{figure}[H]
  \centering
  \includegraphics[width=0.7\linewidth]{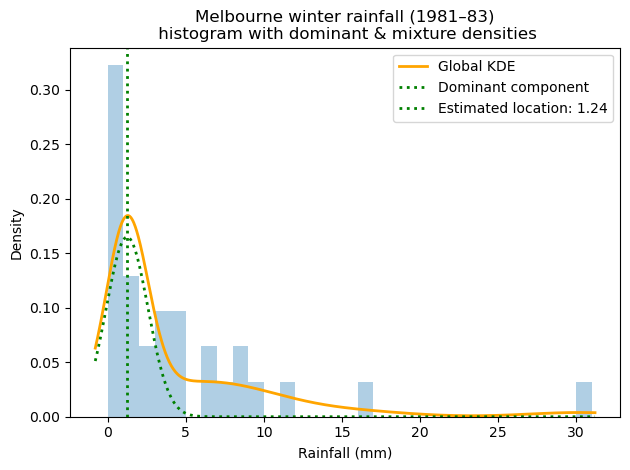} 
  \caption{Rainfall data. Estimated mean 1.239 mm}
  \label{exp8}
\end{figure}


Now let us assess the performance of our algorithm for sequential extraction of multiple significant peaks. As before, we run our algorithm on both simulated and real world datasets.

\subsection{Simulated datasets}
We first simulate $500$ samples from a $3$-component gaussian mixture $0.45\:\mathcal N(-6,1)  + 0.3\;\mathcal{N}(0,2) + 0.25\;\mathcal{N}(8,2.5).$ Our sequential extraction scheme produces 3 peaks with centres $5.972, \; -0.063, \; 7.994$ which are very close to to the true peak locations.

\begin{figure}[H]
  \centering
  \includegraphics[width=0.7\linewidth]{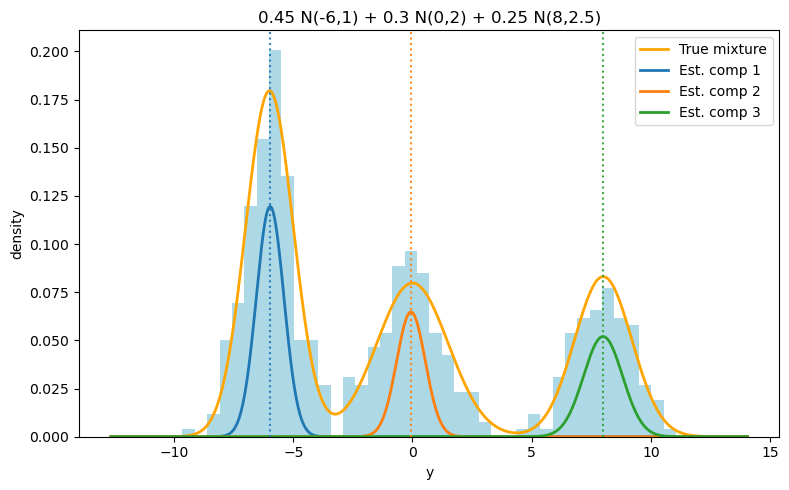} 
  \caption{45-30-25 mixture of N(-6,1), N(0,2), N(8,2.5)}
  \label{exp11}
\end{figure}

To examine a more complicated scenario, where the peaks are not very well separated and with considerable overlap between components, we simulate $1000$ samples from a $3$-component gaussian mixture $0.3\:\mathcal N(-3,1.5)  + 0.4\;\mathcal{N}(0,1) + 0.25\;\mathcal{N}(3,1.5).$ Our sequential extraction scheme produces 3 peaks with centres $-2.866, \; 0.080, \; 2.651$. Remarkably close to the true locations inspite of significant overlap between the individual components.

\begin{figure}[H]
  \centering
  \includegraphics[width=0.7\linewidth]{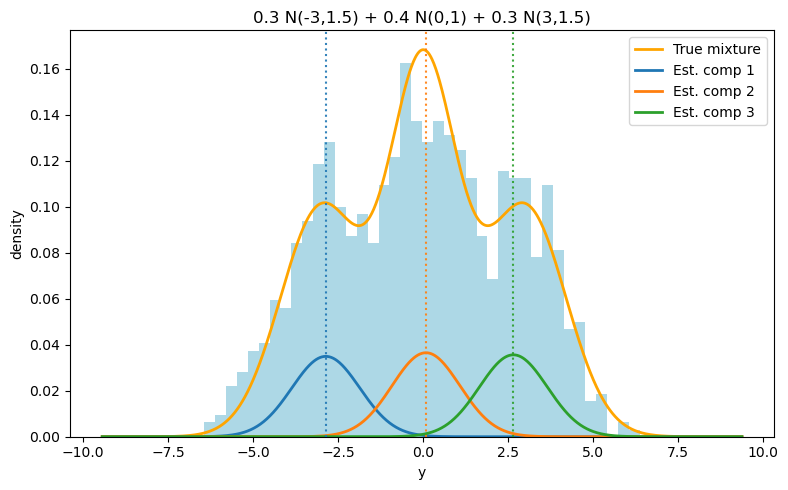} 
  \caption{30-40-30 mixture of N(-3,1.5), N(0,1), N(3,1.5)}
  \label{exp12}
\end{figure}

Now let us see how our algorithm performs where there are even more components. We simulate $10000$ samples from a $5$-component gaussian mixture $0.15\; \mathcal N(-6,1) + 0.20\; \mathcal  N(-3,1) + 0.25\; \mathcal  N(0,1) + 0.20\; \mathcal  N(3,1) + 0.20\; \mathcal  N(6,1).$ Our algorithm produces 5 peaks with centres $-6.184, \; 3.254, \; 0.013, \; 2.994, \; 6.089$, which once again are reasonably close to the true locations.

\begin{figure}[H]
  \centering
  \includegraphics[width=0.7\linewidth]{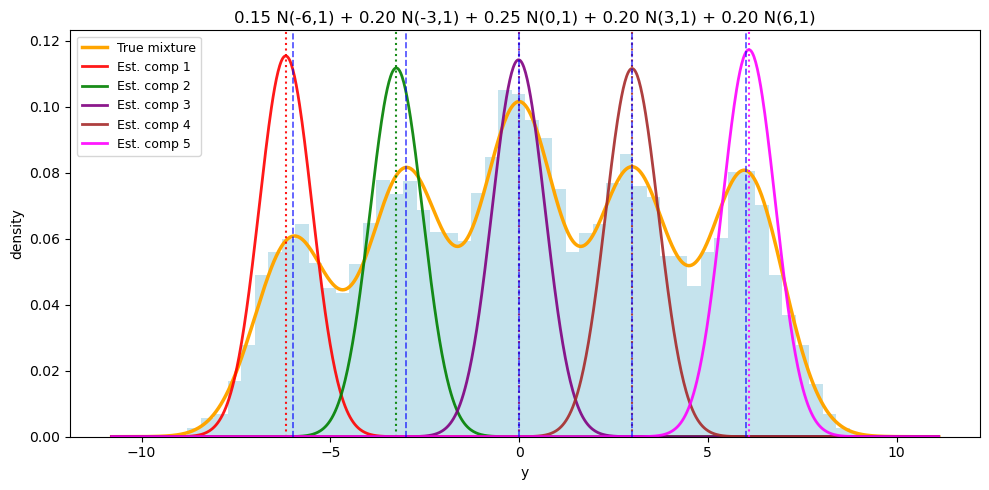} 
  \caption{15-20-25-20-20 mixture of N(-6,1), N(-3,1), N(0,1), N(3,1), N(6,1) }
  \label{exp13}
\end{figure}

\subsection{Real datasets}
The classic Iris data set~\parencite{fisher1936iris}, contains 150 flowers with 50 observations from each of three species (\emph{Iris setosa}, \emph{Iris versicolor}, \emph{Iris virginica}) and four measurements in centimetres (sepal length, sepal width, petal length, petal width); focusing on petal length, the species-wise empirical means and standard deviations are approximately $(\mu,\sigma) = (1.46, 0.17)$ for \emph{setosa}, $(4.26, 0.47)$ for \emph{versicolor}, and $(5.55, 0.55)$ for \emph{virginica}, so that the pooled marginal distribution of petal length is a trimodal mixture with three well-separated peaks corresponding to the three species.
 On applying our sequential algorithm on this dataset, we recover three distinct peaks with centers at $1.449$, $4.493$ and $5.600$, which suggest the existence of three distinct species in the dataset. These estimates are also quite accurate and robust representatives of the empirical means observed in nature even though we have estimated these from a relatively small sample.

\begin{figure}[H]
  \centering
  \includegraphics[width=0.7\linewidth]{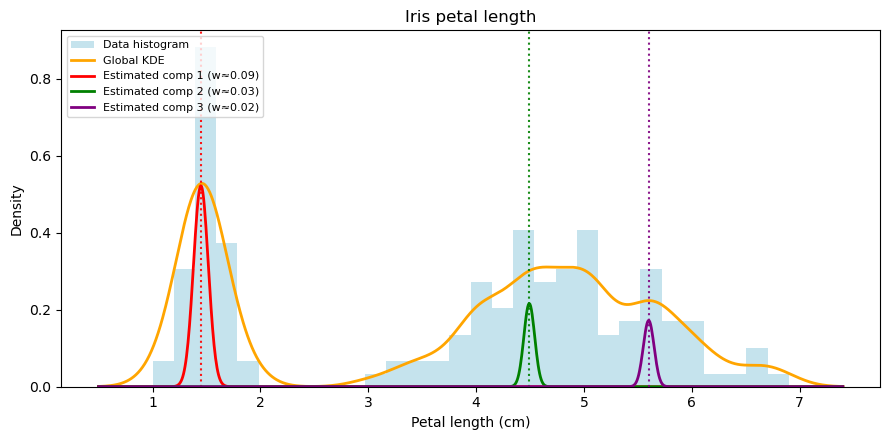} 
  \caption{Multiple component extraction in Iris dataset}
  \label{exp14}
\end{figure}

We now return to the snapper dataset. In the current-year cohort there remains a significant fraction of fish from previous spawning seasons, which are noticeably larger in size. Although the overall population has declined due to migration and fishing practices, a non‑negligible proportion of these older fish is still present in the sample. Applying our sequential scheme to this dataset, we identified two pronounced modes at 
$5.342$ and 
$7.499$, corresponding respectively to the newer and older spawning groups.

\begin{figure}[H]
  \centering
  \includegraphics[width=0.7\linewidth]{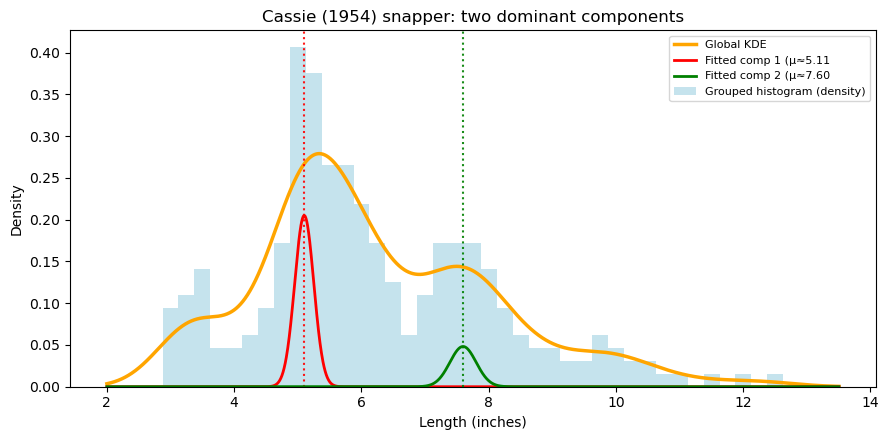} 
  \caption{Snapper data: 2 significant components}
  \label{exp15}
\end{figure}

\appendix
\section*{Proofs and supporting results}
\subsection*{Proof of Lemma~\ref{lem:fft-binning-interp-0based} }
\paragraph{Binning error.}
Fix $y\in\R$ and set $g_y(t):=\mathbb K_{h_n}(y-t)$.
Then $g_y\in C^2$ with $\|g_y''\|_\infty=\|\mathbb K''\|_\infty/h_n^3$.
Since $Y_i=x_{j(i)}+\alpha_i\Delta_n\in[x_{j(i)},x_{j(i)+1}]$, the $C^2$ linear--interpolation remainder gives
\[
\big|g_y(Y_i)-(1-\alpha_i)g_y(x_{j(i)})-\alpha_i g_y(x_{j(i)+1})\big|
\ \le\ \frac{\Delta_n^2}{8}\,\frac{\|\mathbb K''\|_\infty}{h_n^3}.
\]
Multiplying by $\omega_i\ge0$ and summing over $i$ (using $\sum_i\omega_i=1$) yields
\[
\big\|\tilde f^{\mathrm{bin}}_{1,n}-\tilde f^{\mathrm{exact}}_{1,n}\big\|_\infty
\ \le\ \frac{1}{8}\,\|\mathbb K''\|_\infty\,\frac{\Delta_n^2}{h_n^3}.
\tag{B}
\]

\paragraph{Exact FFT on grid.}
With zero padding, the FFT computes the linear convolution on indices $0,\dots,M-1$, so
$\tilde f^{\mathrm{FFT}}_{1,n}(x_m)=\tilde f^{\mathrm{bin}}_{1,n}(x_m)$ for all $m$.

\paragraph{Off--grid interpolation at $Y_i$.}
Let $F(y):=\tilde f^{\mathrm{bin}}_{1,n}(y)=\sum_{m=0}^{M-1} c[m]\mathbb K_{h_n}(y-x_m)$.
Then $F\in C^2$ with
\[
F''(y)=\sum_{m=0}^{M-1} c[m]\,\mathbb K_{h_n}''(y-x_m)
\quad\Rightarrow\quad
\|F''\|_\infty\le \sum_m c[m]\|\mathbb \K_{h_n}''\|_\infty=\frac{\|\mathbb K''\|_\infty}{h_n^3}.
\]
On $[x_{j(i)},x_{j(i)+1}]$, the linear interpolant at $Y_i$ equals
\[
\widehat F(Y_i)=(1-\alpha_i)F(x_{j(i)})+\alpha_i F(x_{j(i)+1}),
\]
with interpolation error
\[
|F(Y_i)-\widehat F(Y_i)|\ \le\ \frac{\Delta_n^2}{8}\,\frac{\|\mathbb K''\|_\infty}{h_n^3}.
\tag{E}
\]
Because $\tilde f^{\mathrm{FFT}}_{1,n}(x_m)=F(x_m)$, our definition
$\tilde f^{\mathrm{FFT}}_{1,n}(Y_i)=\widehat F(Y_i)$ gives
\[
|\tilde f^{\mathrm{FFT}}_{1,n}(Y_i)-\tilde f^{\mathrm{bin}}_{1,n}(Y_i)|\ \le\ \frac{1}{8}\,\|\mathbb K''\|_\infty\,\frac{\Delta_n^2}{h_n^3}.
\]

\paragraph{Triangle inequality.}
Combine (B) and (E):
\[
|\tilde f^{\mathrm{FFT}}_{1,n}(Y_i)-\tilde f^{\mathrm{exact}}_{1,n}(Y_i)|
\ \le\ \frac{1}{8}\,\|\mathbb K''\|_\infty\,\frac{\Delta_n^2}{h_n^3}
      + \frac{1}{8}\,\|\mathbb K''\|_\infty\,\frac{\Delta_n^2}{h_n^3}
\ =\ \frac{1}{4}\,\|\mathbb K''\|_\infty\,\frac{\Delta_n^2}{h_n^3}.
\]
Taking $\max_i$ or $\sup_y$ gives the claim.

\subsection*{Proof of Proposition~\ref{prop:fft-em-iterates-final}}
Introduce latent indicators $Z_i\in\{0,1\}$ ($Z_i=1$ for the parametric component).
The complete--data log--likelihood is
\[
\ell_c(\pi_0,\theta,f_1)=\sum_{i=1}^n\Big\{
Z_i[\log\pi_0+\log f_0(Y_i;\theta)]
+(1-Z_i)[\log(1-\pi_0)+\log f_1(Y_i)]\Big\}.
\]
Given $(\theta^{(t)},\pi_0^{(t)})$ and the current FFT background
$\tilde f^{\mathrm{FFT},(t)}_{1,n}$, the E--step computes
\[
r_i^{(t)}=\E[Z_i\mid Y_i]=
\frac{\pi_0^{(t)} f_0(Y_i;\theta^{(t)})}
     {\pi_0^{(t)} f_0(Y_i;\theta^{(t)})+(1-\pi_0^{(t)})\tilde f^{\mathrm{FFT},(t)}_{1,n}(Y_i)},
\quad s_i^{(t)}=1-r_i^{(t)}.
\]
Taking conditional expectation of $\ell_c$ yields
\[
Q^{(t)}(\pi_0,\theta,f_1)=\frac1n\sum_{i=1}^n
\Big\{ r_i^{(t)}\,[\log\pi_0+\log f_0(Y_i;\theta)]
+s_i^{(t)}\,[\log(1-\pi_0)+\log f_1(Y_i)]\Big\}.
\]
Maximizing $Q^{(t)}$ in $\pi_0\in(0,1)$ with $\pi_1=1-\pi_0$ gives
$\sum_i\big(\frac{r_i^{(t)}}{\pi_0}-\frac{s_i^{(t)}}{1-\pi_0}\big)=0$, hence
$\pi_0^{(t+1)}=\frac1n\sum_i r_i^{(t)}$.
Maximizing in $\theta$ gives the weighted MLE
$\theta^{(t+1)}\in\arg\max_\theta \frac1n\sum_i r_i^{(t)}\log f_0(Y_i;\theta)$,
equivalently the weighted score equation stated.

For the nonparametric part, we adopt the sieve class
$\mathcal F_{h_n,\Delta_n}=\{y\mapsto \sum_{m=0}^{M-1} c[m]\K_{h_n}(y-x_m): c[m]\ge0,\sum_m c[m]=1\}$
and construct the next background by smoothing the normalized nonparametric responsibilities:
first map $\{\omega_i^{(t)}\}$ to grid coefficients $c^{(t)}[\cdot]$ via linear binning, then
compute the grid convolution by zero--padded FFT, and finally use linear interpolation at $Y_i$
in the next E--step. This completes the EM iteration with the FFT plug--in background.

\subsection*{Consistency proof and supporting lemmas}
\begin{assumption}\label{ass:model-id}
\leavevmode
\begin{enumerate}\setlength\itemsep{2pt}
\item[(i)] $f(y):=\pi_0 f_0(y;\theta)+\pi_1 f_1(y)$ satisfies
$f(y)>0$ for all $y\in\R$.
\item[(ii)] $(y,\theta)\mapsto f_0(y;\theta)$ is continuous on $\R\times\mathcal C$.
\item[(iii)] For every $M<\infty$ there exists $c(M)>0$ such that
\[
\inf_{\substack{|y|\le M\\ \theta\in\mathcal C}} f_0(y;\theta)\;\ge\; c(M).
\]
\item[(iv)] There exists $B<\infty$ with
\(
\sup_{\theta\in\mathcal C,\,y} f_0(y;\theta)\le B
\)
and
\(
\sup_{y} f_1^*(y)\le B.
\)
\item[(v)] The population log–likelihood
\[
L(\theta):=\E\Big[\log\{\pi_0 f_0(Y;\theta)+\pi_1 f_1(Y)\}\Big],\qquad
\theta\in\mathcal C,
\]
is continuous on $\mathcal C$ and uniquely maximised at $\theta^*$.
\item[(vi)] There exists $H:\R\to[0,\infty)$ with
$\E[H(Y)]<\infty$ such that
\[
\sup_{\theta\in\mathcal C}
\big|\log\{\pi_0 f_0(Y;\theta)+\pi_1 f_1^*(Y)\}\big|
\;\le\; H(Y)\quad a.s.
\]
\end{enumerate}
\end{assumption}


\begin{assumption}\label{ass:K-h-delta}
$\mathbb K$ is even, bounded, Lipschitz and $C^2$, with
\[
\int \mathbb K=1,\qquad \int u \mathbb K(u)\,du=0,\qquad \int u^2| \mathbb K(u)|\,du<\infty,
\qquad \| \mathbb K''\|_\infty<\infty.
\]
Let $\K_h(u):=h^{-1}\K(u/h)$ and assume
\[
h_n\downarrow0,
\qquad
\frac{n h_n}{\log(1/h_n)}\ \longrightarrow\ \infty .
\]
For the FFT/binned implementation, use a uniform grid of spacing $\Delta_n\downarrow0$ and define
\[
\varepsilon_n\;:=\;\frac{\| \mathbb K''\|_\infty}{4}\,\frac{\Delta_n^2}{h_n^3}\ \longrightarrow\ 0.
\]
Let $\tilde f_{1,n}^{\mathrm{exact}}$ be the weighted KDE and
$\tilde f_{1,n}^{\mathrm{FFT}}$ its binned/FFT approximation. Assume the standard $C^2$ bound
\[
\sup_{y\in\R}\big|\tilde f_{1,n}^{\mathrm{FFT}}(y)-\tilde f_{1,n}^{\mathrm{exact}}(y)\big|
\;\le\; \varepsilon_n.
\]
\end{assumption}


\begin{assumption}\label{ass:kde-sup}
There exist nonnegative (possibly data–dependent) weights $\hat\omega_1,\dots,\hat\omega_n$ with
$\sum_{i=1}^n\hat\omega_i=1$ such that
\[
\tilde f_{1,n}^{\mathrm{exact}}(y)
\;:=\;\sum_{i=1}^n \hat\omega_i\,\mathbb K_{h_n}(y-Y_i),
\qquad y\in\R.
\]
There exists a deterministic $D_n\uparrow\infty$ for which
\[
\sup_{|y|\le D_n}
\big|\tilde f_{1,n}^{\mathrm{exact}}(y)-f_1^*(y)\big|
\;=\;
O_p\!\Big(h_n^2+\sqrt{\tfrac{\log(1/h_n)}{n h_n}}\Big),
\]
and hence, by Assumption~\ref{ass:K-h-delta},
\[
\sup_{|y|\le D_n}
\big|\tilde f_{1,n}^{\mathrm{FFT}}(y)-f_1^*(y)\big|
\;=\;
O_p\!\Big(h_n^2+\sqrt{\tfrac{\log(1/h_n)}{n h_n}}+\varepsilon_n\Big).
\]

\emph{Remark.} Such an assumption is common in the literature, for example -- \cite{Hansen2008}, \cite{GineGuillou2002}, \cite{EinmahlMason2005}, \cite{Tsybakov2009}, \cite{Jiang2017}.
\end{assumption}

For convenience, write
\[
r_n\;:=\;h_n^2+\sqrt{\frac{\log(1/h_n)}{n h_n}}+\varepsilon_n.
\]


\begin{assumption}\label{ass:window-nocf}
For the $D_n$ in Assumption~\ref{ass:kde-sup},
\[
\frac{r_n}{c(D_n)}\ \longrightarrow\ 0,
\qquad
\Pr(|Y|>D_n)\,\log(1/h_n)\ \longrightarrow\ 0.
\]
\end{assumption}

\begin{assumption}\label{ass:log-f0-envelope}
There exists an $H_0:\R\to[0,\infty)$ with $\E[H_0(Y)]<\infty$ such that
\[
\sup_{\theta\in\mathcal C} \big|\log f_0(Y;\theta)\big| \;\le\; H_0(Y)\quad a.s.
\]
\end{assumption}


\begin{lemma}\label{lem:ULLN-nocf}
Under Assumption~\ref{ass:model-id},
\[
\sup_{\theta\in\mathcal C}
\left|
\frac1n\sum_{i=1}^n
\log\{\pi_0 f_0(Y_i;\theta)+\pi_1 f_1^*(Y_i)\}
- \mathcal L(\theta)
\right|
\ \xrightarrow{p}\ 0.
\]
\end{lemma}

\begin{proof}
For each $y$, $\theta\mapsto\log\{\pi_0 f_0(y;\theta)+\pi_1 f_1^*(y)\}$ is continuous on
compact $\mathcal C$ and dominated by $H(Y)$ with $\E H(Y)<\infty$; see
\cite[Thm~19.4]{vaart1998}.
\end{proof}

\begin{proposition}[Uniform convergence of the FFT plug–in criterion]
\label{prop:UC-nocf}
We work under Assumptions~\ref{ass:model-id}, \ref{ass:K-h-delta}, \ref{ass:kde-sup},
\ref{ass:window-nocf} and \ref{ass:log-f0-envelope}.
Let
\[
\tilde {\mathcal{L}_n}(\theta)
:=\frac1n\sum_{i=1}^n
\log\!\Big\{\pi_0 f_0(Y_i;\theta)+\pi_1\,\tilde f_{1,n}^{\mathrm{FFT}}(Y_i)\Big\},
\qquad \theta\in\mathcal C,
\]
and
\(
L(\theta):=\E\log\{\pi_0 f_0(Y;\theta)+\pi_1 f_1^*(Y)\}.
\)
Then
\[
\sup_{\theta\in\mathcal C}\big|\tilde {\mathcal{L}}_n(\theta)-\mathcal{L}(\theta)\big|\ \xrightarrow{p}\ 0.
\]
\end{proposition}

\begin{proof}
Write $\ell_\theta^u(y):=\log\{\pi_0 f_0(y;\theta)+\pi_1 u(y)\}$,
and decompose
\[
\tilde L_n(\theta)-L(\theta)
= R_{1,n}(\theta)+R_{2,n}(\theta),
\qquad
R_{1,n}(\theta):=\frac1n\sum_{i=1}^n\big(\ell_\theta^{\tilde f}(Y_i)-\ell_\theta^{f_1^*}(Y_i)\big),
\quad
R_{2,n}(\theta):=\frac1n\sum_{i=1}^n\ell_\theta^{f_1^*}(Y_i)-L(\theta),
\]

with $\tilde f:=\tilde f_{1,n}^{\mathrm{FFT}}$.
By Lemma~\ref{lem:ULLN-nocf} (oracle ULLN), $\sup_\theta|R_{2,n}(\theta)|\to_p 0$.
We control $R_{1,n}$ via a body/tail split at the deterministic window $M_n\uparrow\infty$
from Assumption~\ref{ass:window-nocf}:
\[
R_{1,n}(\theta)
=
\underbrace{\frac1n\sum_{i=1}^n
\big(\ell_\theta^{\tilde f}(Y_i)-\ell_\theta^{f_1^*}(Y_i)\big)\1_{\{|Y_i|\le M_n\}}}_{=:A_{n,M_n}(\theta)}
+
\underbrace{\frac1n\sum_{i=1}^n
\big(\ell_\theta^{\tilde f}(Y_i)-\ell_\theta^{f_1^*}(Y_i)\big)\1_{\{|Y_i|> M_n\}}}_{=:B_{n,M_n}(\theta)}.
\]

\emph{Body:}
On $\{|Y_i|\le M_n\}$, Assumption~\ref{ass:model-id}(iii) gives
$\pi_0 f_0(Y_i;\theta)\ge \pi_0 c(M_n)$, hence by the mean value theorem for $\log$,
\[
\big|\ell_\theta^{\tilde f}(Y_i)-\ell_\theta^{f_1^*}(Y_i)\big|
\le \frac{\pi_1}{\pi_0 c(M_n)}\,\big|\tilde f(Y_i)-f_1^*(Y_i)\big|.
\]

Therefore, using the trivial averaging inequality
\(
\frac1n\sum_{i=1}^n |a_i|\1_{\{|Y_i|\le M_n\}}
\le \sup_{|y|\le M_n}|a(y)|
\)

with $a(y)=\tilde f(y)-f_1^*(y)$, we obtain
\[
\sup_{\theta\in\mathcal C}|A_{n,M_n}(\theta)|
\le \frac{\pi_1}{\pi_0 c(M_n)}\,\sup_{|y|\le M_n}\big|\tilde f(y)-f_1^*(y)\big|
= O_p\!\Big(\frac{r_n}{c(M_n)}\Big)
= o_p(1),
\]

by Assumptions~\ref{ass:kde-sup} and \ref{ass:window-nocf}.

\emph{Tail:}
Fix $\theta$ and $y$. Since $\tilde f(y) \ge 0$, we have the upper bound
\[
\ell_\theta^{\tilde f}(y)
=\log\{\pi_0 f_0(y;\theta)+\pi_1 \tilde f(y)\}
\le \log\{\pi_0 B+\pi_1(\| \mathbb K\|_\infty h_n^{-1}+\varepsilon_n)\}
\le C_0+\log(1/h_n)
\]

for all large $n$ by Assumptions~\ref{ass:model-id}(iv) and \ref{ass:K-h-delta}.
We also have the lower bound
\(
\ell_\theta^{\tilde f}(y)\ge \log\{\pi_0 f_0(y;\theta)\}.
\)
Combining,
\[
\big|\ell_\theta^{\tilde f}(y)\big|
\le \max\!\big\{C_0+\log(1/h_n),\, -\log(\pi_0 f_0(y;\theta))\big\}
\le C_1+\log(1/h_n)+\big|\log f_0(y;\theta)\big|,
\]

with $C_1$ absorbing $C_0$ and $|\log\pi_0|$.\\

By Assumption~\ref{ass:log-f0-envelope},
\(
\sup_{\theta\in\mathcal C}|\log f_0(Y;\theta)|\le H_0(Y)
\)

and hence
\[
\sup_{\theta\in\mathcal C}\big|\ell_\theta^{\tilde f}(Y)\big|
\le C_1+\log(1/h_n)+H_0(Y).
\]

For the oracle term, Assumption~\ref{ass:model-id} gives
\(
\sup_{\theta\in\mathcal C}\big|\ell_\theta^{f_1^*}(Y)\big|\le H(Y).
\)
Therefore,
\[
\sup_{\theta\in\mathcal C}|B_{n,M_n}(\theta)|
\le \frac1n\sum_{i=1}^n \Big(C_1+\log(1/h_n)+H_0(Y_i)+H(Y_i)\Big)\1_{\{|Y_i|>M_n\}}.
\]

Take expectations and use LLN and Markov inequality:
\[
\E\!\left[\frac1n\sum_{i=1}^n (H_0(Y_i)+H(Y_i))\1_{\{|Y_i|>M_n\}}\right]
= \E[(H_0(Y)+H(Y))\1_{\{|Y|>M_n\}}]\;\longrightarrow\;0
\]

by dominated convergence since $H_0,H\in L^1(g)$ and $M_n\uparrow\infty$.

Also, by Assumption~\ref{ass:window-nocf},
\(
\Pr(|Y|>M_n)\log(1/h_n)\to 0.
\)
As $n^{-1}\sum_{i=1}^n \1_{\{|Y_i|>M_n\}}=P_n(|Y|>M_n)=\Pr(|Y|>M_n)+o_p(1)$,
we conclude
\[
\sup_{\theta\in\mathcal C}|B_{n,M_n}(\theta)|
= o_p(1).
\]

Combining the body and tail bounds,
\(
\sup_{\theta\in\mathcal C}|R_{1,n}(\theta)|\to_p 0.
\)
Together with the oracle ULLN for $R_{2,n}$ this yields
\(
\sup_{\theta\in\mathcal C}|\tilde {\mathcal{L}_n}(\theta)-\mathcal L(\theta)|\to_p 0.
\)
\end{proof}


\subsection*{Proof of Theorem~\ref{thm:consistency-fft-nocf} (Consistency)}
By Proposition~\ref{prop:UC-nocf},
\(
\sup_{\theta\in\mathcal C}|\tilde {\mathcal{L}}_n(\theta)-\mathcal L(\theta)|\to_p 0.
\)
Assumption~\ref{ass:model-id}(v) yields uniqueness and continuity of the maximiser on compact
$\mathcal C$. Now apply Wald’s argmax theorem.
\\~\\
\subsection*{CLT assumptions and supporting lemmas}
Write
\[
\ell(y;\theta):=\log\{\pi_0 f_0(y;\theta)+\pi_1 f_1^*(y)\},\qquad
\dot\ell(y;\theta):=\nabla_\theta \ell(y;\theta),\qquad
\ddot\ell(y;\theta):=\nabla_\theta^2 \ell(y;\theta),
\]
and
\[
S_n(\theta):=\frac1n\sum_{i=1}^n \dot\ell(Y_i;\theta),\qquad
H_n(\theta):=\frac1n\sum_{i=1}^n \ddot\ell(Y_i;\theta).
\]
We have
\[
\tilde {\mathcal L}_n(\theta)
:=\frac1n\sum_{i=1}^n
\log\!\Big\{\pi_0 f_0(Y_i;\theta)+\pi_1\,\tilde f_{1,n}^{\mathrm{FFT}}(Y_i)\Big\},
\qquad \theta\in\mathcal C,
\]
with maximiser $\hat\theta_n\in\arg\max_{\theta\in\mathcal C}\tilde L_n(\theta)$.
By Theorem~\ref{thm:consistency-fft-nocf}, $\hat\theta_n\to_p\theta^*$.

Recall
\[
r_n:=h_n^2+\sqrt{\frac{\log(1/h_n)}{n h_n}}+\varepsilon_n,
\]
and the window $D_n\uparrow\infty$ from Assumption~\ref{ass:window-nocf}.


\begin{assumption}
\label{ass:diff-oracle-nocf}
For each $y$, $\theta\mapsto f_0(y;\theta)$ is twice continuously differentiable
on $\mathcal C$, and differentiation under the integral sign is valid
for $\ell,\dot\ell,\ddot\ell$. There exist envelopes $H_1,H_2\in L^2(f)$ such that,
for all $\theta\in\mathcal C$,
\[
\|\dot\ell(Y;\theta)\|\le H_1(Y),\qquad
\|\ddot\ell(Y;\theta)\|\le H_2(Y).
\]

Moreover, for each $M<\infty$,
\[
\sup_{\substack{|y|\le M\\ \theta\in\mathcal C}}
\Big(\|\nabla_\theta f_0(y;\theta)\|+\|\nabla_\theta^2 f_0(y;\theta)\|\Big)<\infty.
\]

The Fisher information
\[
I(\theta^*):=\E_f\!\big[\dot\ell(Y;\theta^*)\dot\ell(Y;\theta^*)^\top\big]
=-\,\E_f\!\big[\ddot\ell(Y;\theta^*)\big]
\]

exists, is finite, and is positive definite.
\end{assumption}

\begin{assumption}
\label{ass:plugin-tail-nocf}
There exist measurable functions $H_{1,n},H_{2,n}:\R\to[0,\infty)$ such that,
for all $\theta\in\mathcal C$,
\[
\Big\|\nabla_\theta\log\{\pi_0 f_0(Y;\theta)+\pi_1 \tilde f_{1,n}^{\mathrm{FFT}}(Y)\}\Big\|
\le H_{1,n}(Y),
\]
\[
\Big\|\nabla_\theta^2\log\{\pi_0 f_0(Y;\theta)+\pi_1 \tilde f_{1,n}^{\mathrm{FFT}}(Y)\}\Big\|
\le H_{2,n}(Y),
\]
and, with $D_n$ as in Assumption~\ref{ass:window-nocf},
\[
\sqrt{n}\,\E_f\!\big[(H_{1,n}(Y)+H_1(Y))\,\1_{\{|Y|>D_n\}}\big]\ \longrightarrow\ 0,
\qquad
\E_f\!\big[(H_{2,n}(Y)+H_2(Y))\,\1_{\{|Y|>D_n\}}\big]\ \longrightarrow\ 0.
\]
\end{assumption}

\begin{assumption}
\label{ass:targeted-nocf}
Let
\[
\Phi(y):=\frac{\pi_0\pi_1\,\nabla_\theta f_0(y;\theta^*)}
               {\{\pi_0 f_0(y;\theta^*)+\pi_1 f_1^*(y)\}^2}\in\R,
\]
and define the \emph{windowed} linear functional
\[
\Lambda_n:=\frac1n\sum_{i=1}^n \Phi(Y_i)\,
\big(\tilde f_{1,n}^{\mathrm{FFT}}(Y_i)-f_1^*(Y_i)\big)\,\1_{\{|Y_i|\le D_n\}}\in\R .
\]
Then $\sqrt{n}\,\|\Lambda_n\|\to_p 0$.
\end{assumption}

\begin{assumption}
\label{ass:remainder-sep-nocf}
With $r_n$ and $D_n$ as above,
\[
\frac{\sqrt{n}\,r_n^2}{c(D_n)^3}\ \longrightarrow\ 0,
\]
where $c(M)$ is from Assumption~\ref{ass:model-id}(iii).
\end{assumption}


\begin{lemma}[Oracle score CLT]\label{lem:oracle-clt-nocf}
Under Assumptions~\ref{ass:model-id}, and \ref{ass:diff-oracle-nocf},
\[
\sqrt{n}\,S_n(\theta^*)\ \rightsquigarrow\ \mathcal N\big(0,\,I(\theta^*)\big).
\]
\end{lemma}

\begin{proof}
$\dot\ell(Y_i;\theta^*)$ are i.i.d.\ with mean $0$, covariance $I(\theta^*)$,
and $\E\|\dot\ell(Y;\theta^*)\|^2\le\E[H_1(Y)^2]<\infty$. Apply the multivariate
Lindeberg--Feller CLT.
\end{proof}

\begin{lemma}\label{lem:oracle-hess-ulln-nocf}
Under Assumptions \ref{ass:model-id} and ~\ref{ass:diff-oracle-nocf},
\[
\sup_{\theta\in\mathcal C}
\big\|H_n(\theta)-\E_f[\ddot\ell(Y;\theta)]\big\|\ \xrightarrow{p}\ 0.
\]
In particular, by continuity of $\theta\mapsto \E_f[\ddot\ell(Y;\theta)]$
and Assumption~\ref{ass:diff-oracle-nocf}, we have
$H_n(\theta^*)\to_p -I(\theta^*)$.
\end{lemma}

\begin{proof}
For each $y$, $\theta\mapsto\ddot\ell(y;\theta)$ is continuous on compact
$\mathcal C$, and $\|\ddot\ell(Y;\theta)\|\le H_2(Y)$ with
$\E H_2(Y)<\infty$. A dominated Glivenko--Cantelli theorem for parametric
classes (e.g.\ \cite[Thm.~19.4]{vaart1998}) yields the uniform LLN.
Continuity of $\theta\mapsto\E[\ddot\ell(Y;\theta)]$ then gives the
specialization at $\theta^*$.
\end{proof}


\begin{lemma}\label{lem:plugin-diff-nocf}
Let
\[
\tilde S_n(\theta):=\frac1n\sum_{i=1}^n
\nabla_\theta\log\{\pi_0 f_0(Y_i;\theta)+\pi_1 \tilde f_{1,n}^{\mathrm{FFT}}(Y_i)\},\quad
\tilde H_n(\theta):=\frac1n\sum_{i=1}^n
\nabla_\theta^2\log\{\pi_0 f_0(Y_i;\theta)+\pi_1 \tilde f_{1,n}^{\mathrm{FFT}}(Y_i)\}.
\]
Under Assumptions~\ref{ass:model-id},
\ref{ass:K-h-delta}, \ref{ass:kde-sup}, \ref{ass:window-nocf}, \ref{ass:diff-oracle-nocf}, \ref{ass:plugin-tail-nocf},
\ref{ass:targeted-nocf} and \ref{ass:remainder-sep-nocf},
\[
\sqrt{n}\,\big\|\tilde S_n(\theta^*)-S_n(\theta^*)\big\|\ \xrightarrow{p}\ 0,
\qquad
\sup_{\theta\in\mathcal C}\big\|\tilde H_n(\theta)-H_n(\theta)\big\|\ \xrightarrow{p}\ 0.
\]
\end{lemma}

\begin{proof}
Write $\tilde s_i(\theta)$ and $s_i(\theta)$ for the summand scores, and split
with the window $D_n$:
\[
\tilde S_n(\theta^*)-S_n(\theta^*)
=\underbrace{\frac1n\sum_{i=1}^n(\tilde s_i-s_i)\,\1_{\{|Y_i|\le D_n\}}}_{=:B_n}
 +\underbrace{\frac1n\sum_{i=1}^n(\tilde s_i-s_i)\,\1_{\{|Y_i|>D_n\}}}_{=:T_n}.
\]
By Assumption~\ref{ass:plugin-tail-nocf},
\(
\|\tilde s_i-s_i\|\le H_{1,n}(Y_i)+H_1(Y_i)
\),
so
\[
\sqrt{n}\,\|T_n\|
\le \sqrt{n}\,\frac1n\sum_{i=1}^n (H_{1,n}(Y_i)+H_1(Y_i))\,\1_{\{|Y_i|>D_n\}}.
\]
Define
\[
Z_n:=\sqrt{n}\,\frac1n\sum_{i=1}^n (H_{1,n}(Y_i)+H_1(Y_i))\,\1_{\{|Y_i|>D_n\}}.
\]
Then
\(
\E[Z_n]=\sqrt{n}\,\E[(H_{1,n}+H_1)\1_{\{|Y|>D_n\}}]\to0
\)
by Assumption~\ref{ass:plugin-tail-nocf}, so $Z_n=o_p(1)$ by Markov.
Thus $\sqrt{n}\,\|T_n\|=o_p(1)$. The Hessian tail is analogous with
$H_{2,n},H_2$ and no $\sqrt{n}$ factor.

For $|Y_i|\le D_n$, the map
\(
u\mapsto\nabla_\theta\log\{\pi_0 f_0(Y_i;\theta^*)+\pi_1 u(Y_i)\}
\)
is twice Fr\'echet--differentiable near $u=f_1^*$ with first derivative
$D_u[\nabla_\theta\log](y)\cdot v=-\Phi(y)\,v(y)$ and second derivative bounded, on $\{|y|\le D_n\}$, by $C/c(D_n)^3$. With
$\Delta_n(y):=\tilde f_{1,n}^{\mathrm{FFT}}(y)-f_1^*(y)$,
\[
\tilde s_i(\theta^*)-s_i(\theta^*)
=-\Phi(Y_i)\,\Delta_n(Y_i)+R_i,\qquad
\|R_i\|\lesssim \frac{\Delta_n(Y_i)^2}{c(D_n)^3}.
\]

Thus
\[
B_n
= -\,\Lambda_n
  + \frac1n\sum_{i=1}^n R_i\,\1_{\{|Y_i|\le D_n\}}.
\]

By Assumption~\ref{ass:targeted-nocf}, $\sqrt{n}\,\|\Lambda_n\|\to_p0$.
By Assumption~\ref{ass:kde-sup},
$\sup_{|y|\le D_n}|\Delta_n(y)|=O_p(r_n)$, hence
\[
\sqrt{n}\cdot\frac1n\sum_{i=1}^n \|R_i\|\,\1_{\{|Y_i|\le D_n\}}
=O_p\!\Big(\frac{\sqrt{n}\,r_n^2}{c(D_n)^3}\Big)=o_p(1)
\]

by Assumption~\ref{ass:remainder-sep-nocf}. This proves the score claim.
The Hessian claim is obtained by an analogous expansion with second
derivatives.
\end{proof}


\subsection*{Proof of Theorem~\ref{thm:clt-fft-nocf} (CLT)}
Consistency gives $\hat\theta_n\to_p\theta^*$. The first--order condition
$\tilde S_n(\hat\theta_n)=0$ and a Taylor expansion imply
\[
0=\tilde S_n(\theta^*)+\tilde H_n(\tilde\theta_n)(\hat\theta_n-\theta^*),
\]
so
\[
\sqrt{n}(\hat\theta_n-\theta^*)
=-\,\tilde H_n(\tilde\theta_n)^{-1}\,\sqrt{n}\,\tilde S_n(\theta^*).
\]
Decompose
\[
\sqrt{n}\,\tilde S_n(\theta^*)
=\sqrt{n}\,S_n(\theta^*)+\sqrt{n}\,(\tilde S_n(\theta^*)-S_n(\theta^*)).
\]
By Lemma~\ref{lem:oracle-clt-nocf}, 
\[\sqrt{n}\,S_n(\theta^*)\Rightarrow\mathcal N(0,I(\theta^*))\]
By Lemma~\ref{lem:plugin-diff-nocf},
\[\sqrt{n}\,(\tilde S_n(\theta^*)-S_n(\theta^*))=o_p(1)\]
and
\[\sup_{\theta}\|\tilde H_n(\theta)-H_n(\theta)\|\to_p 0\]
By Lemma~\ref{lem:oracle-hess-ulln-nocf} and the continuity of
$\theta\mapsto\E[\ddot\ell(Y;\theta)]$, we also have
$H_n(\theta)\to_p \E[\ddot\ell(Y;\theta)]$ uniformly and, with $\tilde\theta_n\to_p\theta^*$, we have
\(
\tilde H_n(\tilde\theta_n)\to_p -I(\theta^*).
\)
Now Slutsky's theorem yields the claim.
\subsection*{Consistency proof and supporting lemmas (sequential scheme)}


\begin{assumption}
\label{ass:model-id-stage}
\leavevmode
\begin{enumerate}\setlength\itemsep{2pt}
\item[(i)] $f_0(\cdot;\theta)$ is a Lebesgue density for each $\theta\in\mathcal C$,
$f_{\mathrm{bg}}^{(\ell)}$ is a Lebesgue density, and $f^{(\ell)}(y)>0$ for all $y$.
\item[(ii)] $(y,\theta)\mapsto f_0(y;\theta)$ is continuous on
$\R\times\mathcal C$.
\item[(iii)] For every $M<\infty$ there exists $c_\ell(M)>0$ such that
\[
\inf_{\substack{|y|\le M\\ \theta\in\mathcal C}}
f_0(y;\theta)\ge c_\ell(M).
\]
\item[(iv)] There exists $B_\ell<\infty$ with
$\sup_{\theta\in\mathcal C,\,y} f_0(y;\theta)\le B_\ell$ and
$\sup_y f_{\mathrm{bg}}^{(\ell)}(y)\le B_\ell$.
\item[(v)] The oracle criterion $L^{(\ell)}$ is continuous on
$\mathcal C\times\Pi_{\min}$ and uniquely maximised at
$\eta_\ell=(\theta_\ell,\alpha_\ell^{(\ell)})$.
\item[(vi)] (Oracle log--envelope.) There exists $H^{(\ell)}\in L^1(f^{(\ell)})$ such that
\[
\sup_{\theta\in\mathcal C,\;\alpha\in\Pi_{\min}}
\big|\ell^{(\ell)}\big(Y;\eta,f_{\mathrm{bg}}^{(\ell)}\big)\big|
\le H^{(\ell)}(Y)\quad a.s.
\]
\end{enumerate}
\end{assumption}

\begin{assumption}
\label{ass:kde-window-stage}
Let $\K$ be even, bounded, Lipschitz and $C^2$ with
$\int \K=1$, $\int u\K(u)\,du=0$, $\int u^2|\K(u)|\,du<\infty$ and
$\|\K''\|_\infty<\infty$. Define $\K_{h_n}(u):=h_n^{-1}\K(u/h_n)$ and assume
$h_n\downarrow0$, $n h_n/\log(1/h_n)\to\infty$. Let $\tilde f_{\ell,n}^{\mathrm{exact}}$
denote the exact weighted KDE with weights $\widehat w_i^{(\ell)}$:
\[
\tilde f_{\ell,n}^{\mathrm{exact}}(y)
:=\sum_{i=1}^n \widehat w_i^{(\ell)} \mathbb K_{h_n}(y-Y_i),
\]
and $\tilde f_{\ell,n}^{\mathrm{FFT}}$ its FFT approximation on a grid of
spacing $\Delta_n\downarrow0$. Suppose there exists a deterministic window
$M_n\uparrow\infty$ such that
\[
\sup_{|y|\le M_n}
\big|\tilde f_{\ell,n}^{\mathrm{exact}}(y)-f_{\mathrm{bg}}^{(\ell),*}(y)\big|
=O_p\!\Big(h_n^2+\sqrt{\tfrac{\log(1/h_n)}{n h_n}}\Big),
\]
and
\[
\sup_{y\in\R}
\big|\tilde f_{\ell,n}^{\mathrm{FFT}}(y)-\tilde f_{\ell,n}^{\mathrm{exact}}(y)\big|
\le \varepsilon_n,\qquad \varepsilon_n\to0.
\]
In particular, for
\[
r_n:=h_n^2+\sqrt{\tfrac{\log(1/h_n)}{n h_n}}+\varepsilon_n,
\]
we have
\[
\sup_{|y|\le M_n}
\big|\tilde f_{\ell,n}^{\mathrm{FFT}}(y)-f_{\mathrm{bg}}^{(\ell),*}(y)\big|
=O_p(r_n).
\]
\end{assumption}

\begin{assumption}
\label{ass:window-stage}
With $M_n$ and $r_n$ as in Assumption~\ref{ass:kde-window-stage}, we have
\[
\frac{r_n}{c_\ell(M_n)}\ \longrightarrow\ 0,
\qquad
\Pr_{f^{(\ell)}}(|Y|>M_n)\,\log(1/h_n)\ \longrightarrow\ 0.
\]
\end{assumption}

\begin{assumption}
\label{ass:weights-stage}
\leavevmode
\begin{enumerate}\setlength\itemsep{2pt}
\item[(i)] For every measurable $\Phi:\R\to\R$ with
$\E_{f^{(\ell)}}|\Phi(Y)|<\infty$,
\[
P_n^{(\ell)}\Phi
:=\sum_{i=1}^n \widehat w_i^{(\ell)}\Phi(Y_i)
\ \xrightarrow{p}\ \E_{f^{(\ell)}}[\Phi(Y)].
\]
\item[(ii)] There exists $C_\ell<\infty$ such that
\(
0\le \widehat w_i^{(\ell)}\le C_\ell/n
\)
for all $i$ and all large $n$.
\end{enumerate}
\end{assumption}

\begin{assumption}
\label{ass:log-env-plugin-stage}
There exists $H^{(\ell)}_{\mathrm{plug}}\in L^1(f^{(\ell)})$ such that, for all
$\eta\in\mathcal C\times\Pi_{\min}$,
\[
\big|\ell^{(\ell)}\big(Y;\eta,\tilde f_{\ell,n}^{\mathrm{FFT}}\big)\big|
\le H^{(\ell)}_{\mathrm{plug}}(Y)\quad a.s.\ \text{for all large $n$},
\]
and
\(
\E_{f^{(\ell)}}[H^{(\ell)}_{\mathrm{plug}}(Y)\,\mathbf 1_{\{|Y|>M_n\}}]\to0.
\)
\end{assumption}


\begin{proposition}
\label{prop:UC-stage}
Under Assumptions~\ref{ass:model-id-stage}--\ref{ass:log-env-plugin-stage},
\[
\sup_{\eta\in\mathcal C\times\Pi_{\min}}
\big|\tilde {\mathcal{L}}_n^{(\ell)}(\eta) - \mathcal L^{(\ell)}(\eta)\big|
\ \xrightarrow{p}\ 0.
\]
\end{proposition}

\begin{proof}
Fix $\eta=(\theta,\alpha)$ and abbreviate
\(
\ell_\eta^u(y):=\ell^{(\ell)}(y;\eta,u).
\)

Write
\[
\tilde {\mathcal{L}}_n^{(\ell)}(\eta)-\mathcal L^{(\ell)}(\eta)
=R_{1,n}(\eta)+R_{2,n}(\eta),
\]

where
\[
R_{1,n}(\eta)
:=P_n^{(\ell)}\big(\ell_\eta^{\tilde f_{\ell,n}^{\mathrm{FFT}}}
                 -\ell_\eta^{f_{\mathrm{bg}}^{(\ell),*}}\big),
\quad
R_{2,n}(\eta)
:=P_n^{(\ell)}\big(\ell_\eta^{f_{\mathrm{bg}}^{(\ell),*}}\big)
  -\E_{f^{(\ell)}}\big[\ell_\eta^{f_{\mathrm{bg}}^{(\ell),*}}(Y)\big].
\]

\emph{Control of $R_{1,n}$.}
Let $M_n$ be as in Assumption~\ref{ass:kde-window-stage} and write
$R_{1,n}=A_{n,M_n}+B_{n,M_n}$ with
\[
A_{n,M_n}(\eta)
:=P_n^{(\ell)}\big(\ell_\eta^{\tilde f_{\ell,n}^{\mathrm{FFT}}}
                  -\ell_\eta^{f_{\mathrm{bg}}^{(\ell),*}}\big)\mathbf 1_{\{|Y|\le M_n\}},
\]
\[
B_{n,M_n}(\eta)
:=P_n^{(\ell)}\big(\ell_\eta^{\tilde f_{\ell,n}^{\mathrm{FFT}}}
                  -\ell_\eta^{f_{\mathrm{bg}}^{(\ell),*}}\big)\mathbf 1_{\{|Y|>M_n\}}.
\]

On $\{|Y|\le M_n\}$, Assumption~\ref{ass:model-id-stage}(iii) and
$\alpha\in\Pi_{\min}$ imply
\(
\alpha f_0(Y;\theta)\ge \underline\alpha_\ell,c_\ell(M_n).
\)
By the mean value theorem,
\[
\big|\ell_\eta^{\tilde f_{\ell,n}^{\mathrm{FFT}}}(Y)
     -\ell_\eta^{f_{\mathrm{bg}}^{(\ell),*}}(Y)\big|
\le \frac{1-\underline\alpha_\ell}{\underline\alpha_\ell,c_\ell(M_n)}
     \big|\tilde f_{\ell,n}^{\mathrm{FFT}}(Y)-f_{\mathrm{bg}}^{(\ell),*}(Y)\big|
\]

on $\{|Y|\le M_n\}$. Using the weighted analogue of the trivial averaging
inequality,
\[
\Big|\sum_{i=1}^n \widehat w_i^{(\ell)}
        \Delta(Y_i)\mathbf 1_{\{|Y_i|\le M_n\}}\Big|
\le \sup_{|y|\le M_n}|\Delta(y)|,
\]

with $\Delta(y)=\tilde f_{\ell,n}^{\mathrm{FFT}}(y)-f_{\mathrm{bg}}^{(\ell),*}(y)$,
we obtain
\[
\sup_{\eta}|A_{n,M_n}(\eta)|
\le \frac{1-\underline\alpha_\ell}{\underline\alpha_\ell,c_\ell(M_n)}
    \sup_{|y|\le M_n}
    \big|\tilde f_{\ell,n}^{\mathrm{FFT}}(y)-f_{\mathrm{bg}}^{(\ell),*}(y)\big|
=O_p\!\Big(\frac{r_n}{c_\ell(M_n)}\Big)=o_p(1),
\]

by Assumptions~\ref{ass:kde-window-stage} and \ref{ass:window-stage}.

For the tail $B_{n,M_n}$, note that
\[
\big|\ell_\eta^{\tilde f_{\ell,n}^{\mathrm{FFT}}}(Y)
     -\ell_\eta^{f_{\mathrm{bg}}^{(\ell),*}}(Y)\big|
\le \big|\ell_\eta^{\tilde f_{\ell,n}^{\mathrm{FFT}}}(Y)\big|
    +\big|\ell_\eta^{f_{\mathrm{bg}}^{(\ell),*}}(Y)\big|,
\]

and that
\(
\sup_{\eta,y}\ell_\eta^{\tilde f_{\ell,n}^{\mathrm{FFT}}}(y)
\le C_0+\log(1/h_n)
\)
for all large $n$ by Assumptions~\ref{ass:model-id-stage}(iv) and
\ref{ass:kde-window-stage}, while
\(
\sup_{\eta}|\ell_\eta^{f_{\mathrm{bg}}^{(\ell),*}}(Y)|
\le H^{(\ell)}(Y)
\)
by Assumption~\ref{ass:model-id-stage}(vi).
Using $\sum_i\widehat w_i^{(\ell)}\le C_\ell$,
\[
\sup_{\eta}|B_{n,M_n}(\eta)|
\lesssim \big(C_0+\log(1/h_n)\big)\,P_n^{(\ell)}(|Y|>M_n)
       + P_n^{(\ell)}\big(H^{(\ell)}(Y)\mathbf 1_{\{|Y|>M_n\}}\big).
\]

By Assumption~\ref{ass:weights-stage}(i) applied to indicators and $H^{(\ell)}$,
and Assumption~\ref{ass:window-stage},
\[
\big(C_0+\log(1/h_n)\big)\,P_n^{(\ell)}(|Y|>M_n)
= o_p(1),\qquad
P_n^{(\ell)}\big(H^{(\ell)}(Y)\mathbf 1_{\{|Y|>M_n\}}\big)
\to_p 0.
\]

Thus $\sup_{\eta}|B_{n,M_n}(\eta)|=o_p(1)$.

\emph{Control of $R_{2,n}$.}
By Assumptions~\ref{ass:model-id-stage}(vi) and \ref{ass:weights-stage}(i),
for each fixed $\eta$,
\[
R_{2,n}(\eta)=P_n^{(\ell)}\big(\ell_\eta^{f_{\mathrm{bg}}^{(\ell),*}}\big)
             -\E_{f^{(\ell)}}\big[\ell_\eta^{f_{\mathrm{bg}}^{(\ell),*}}(Y)\big]
\to_p0.
\]

Continuity of $\eta\mapsto\ell_\eta^{f_{\mathrm{bg}}^{(\ell),*}}(y)$, compactness of
$\mathcal C\times\Pi_{\min}$, and the envelope $H^{(\ell)}$ imply,
by a finite $\varepsilon$--net argument, that
\(
\sup_{\eta}|R_{2,n}(\eta)|\to_p 0.
\)

Combining these bounds and taking the supremum over $\eta$ yields the claim.
\end{proof}
\subsection*{Proof of Theorem~\ref{thm:consistency-stage} (stage-$\ell$ consistency)}
By Proposition~\ref{prop:UC-stage},
\(
\sup_{\eta\in\mathcal C\times\Pi_{\min}}
|\tilde {\mathcal{L}}_n^{(\ell)}(\eta) - \mathcal L^{(\ell)}(\eta)|\to_p0.
\)
By Assumption~\ref{ass:model-id-stage}(v),
$L^{(\ell)}$ is continuous on the compact set
$\mathcal C\times\Pi_{\min}$ with a unique maximiser $\eta_\ell^*$.
Wald's argmax theorem yields $\hat\eta_{\ell,n}\to_p\eta_\ell^*$.

\subsection*{CLT proof and supporting lemmas}
\begin{assumption}
\label{ass:diff-info-stage}
For each $y$, $\eta\mapsto\ell^{(\ell)}(y;\eta,f_{\mathrm{bg}}^{(\ell),*})$ is twice
continuously differentiable on $\mathcal C\times\Pi_{\min}$ and
differentiation under the integral sign is valid. Define the oracle score and
Hessian
\[
S_n^{(\ell)}(\eta)
:=\sum_{i=1}^n \widehat w_i^{(\ell)}
  \nabla_\eta\ell^{(\ell)}\big(Y_i;\eta,f_{\mathrm{bg}}^{(\ell),*}\big),
\quad
H_n^{(\ell)}(\eta)
:=\sum_{i=1}^n \widehat w_i^{(\ell)}
  \nabla_\eta^2\ell^{(\ell)}\big(Y_i;\eta,f_{\mathrm{bg}}^{(\ell),*}\big).
\]

There exist envelopes $H_{1}^{(\ell)},H_{2}^{(\ell)}\in L^2(f^{(\ell)})$ with
\[
\|\nabla_\eta\ell^{(\ell)}(Y;\eta,f_{\mathrm{bg}}^{(\ell),*})\|\le H_{1}^{(\ell)}(Y),\quad
\|\nabla_\eta^2\ell^{(\ell)}(Y;\eta,f_{\mathrm{bg}}^{(\ell),*})\|
\le H_{2}^{(\ell)}(Y)
\]

for all $\eta\in\mathcal C\times\Pi_{\min}$, and
\[
\sup_{\eta\in\mathcal C\times\Pi_{\min}}
\big\|H_n^{(\ell)}(\eta)
    -\E_{f^{(\ell)}}[\nabla_\eta^2\ell^{(\ell)}(Y;\eta,f_{\mathrm{bg}}^{(\ell),*})]\big\|
\ \xrightarrow{p}\ 0.
\]

Moreover, the information matrix
\[
I^{(\ell)}(\eta_\ell^*)
:=\E_{f^{(\ell)}}\big[
   \nabla_\eta\ell^{(\ell)}(Y;\eta_\ell^*)
   \nabla_\eta\ell^{(\ell)}(Y;\eta_\ell^*)^\top
 \big]
=-\,\E_{f^{(\ell)}}\big[\nabla_\eta^2\ell^{(\ell)}(Y;\eta_\ell^*)\big]
\]

exists, is finite, and is positive definite.
\end{assumption}

\begin{assumption}
\label{ass:score-clt-stage}
Under the weights $\widehat w_i^{(\ell)}$ and $f^{(\ell)}$, the oracle score
satisfies
\[
\sqrt{n}\,S_n^{(\ell)}(\eta_\ell^*)
\ \rightsquigarrow\ \mathcal N\big(0,I^{(\ell)}(\eta_\ell^*)\big).
\]
\end{assumption}

\begin{assumption}
\label{ass:plugin-oracle-stage}
Let
\[
\tilde S_n^{(\ell)}(\eta)
:=\sum_{i=1}^n \widehat w_i^{(\ell)}
  \nabla_\eta\ell^{(\ell)}\big(Y_i;\eta,\tilde f_{\ell,n}^{\mathrm{FFT}}\big),
\quad
\tilde H_n^{(\ell)}(\eta)
:=\sum_{i=1}^n \widehat w_i^{(\ell)}
  \nabla_\eta^2\ell^{(\ell)}\big(Y_i;\eta,\tilde f_{\ell,n}^{\mathrm{FFT}}\big).
\]

Assume
\[
\sqrt{n}\,\big\|\tilde S_n^{(\ell)}(\eta_\ell^*)-S_n^{(\ell)}(\eta_\ell^*)\big\|
\ \xrightarrow{p}\ 0,
\qquad
\sup_{\eta\in\mathcal C\times\Pi_{\min}}
\big\|\tilde H_n^{(\ell)}(\eta)-H_n^{(\ell)}(\eta)\big\|
\ \xrightarrow{p}\ 0.
\]
\end{assumption}

\subsection*{Proof of Theorem~\ref{thm:clt-stage} (stage-$\ell$ CLT)}
By Theorem~\ref{thm:consistency-stage}, $\hat\eta_{\ell,n}\to_p\eta_\ell^*$.
The first--order condition for the maximiser is
$\tilde S_n^{(\ell)}(\hat\eta_{\ell,n})=0$. Let $\tilde\eta_{\ell,n}$ lie on
the line segment between $\hat\eta_{\ell,n}$ and $\eta_\ell^*$. A Taylor
expansion of $\tilde S_n^{(\ell)}$ around $\eta_\ell^*$ gives
\[
0=\tilde S_n^{(\ell)}(\eta_\ell^*)
  +\tilde H_n^{(\ell)}(\tilde\eta_{\ell,n})(\hat\eta_{\ell,n}-\eta_\ell^*),
\]
so
\[
\sqrt{n}(\hat\eta_{\ell,n}-\eta_\ell^*)
=-\,\tilde H_n^{(\ell)}(\tilde\eta_{\ell,n})^{-1}
   \sqrt{n}\,\tilde S_n^{(\ell)}(\eta_\ell^*).
\]

Decompose
\[
\sqrt{n}\,\tilde S_n^{(\ell)}(\eta_\ell^*)
=\sqrt{n}\,S_n^{(\ell)}(\eta_\ell^*)
 +\sqrt{n}\,\big(\tilde S_n^{(\ell)}(\eta_\ell^*)-S_n^{(\ell)}(\eta_\ell^*)\big).
\]

By Assumption~\ref{ass:score-clt-stage},
$\sqrt{n}\,S_n^{(\ell)}(\eta_\ell^*)\rightsquigarrow
 \mathcal N(0,I^{(\ell)}(\eta_\ell^*))$.
 
By Assumption~\ref{ass:plugin-oracle-stage},
$\sqrt{n}(\tilde S_n^{(\ell)}(\eta_\ell^*)-S_n^{(\ell)}(\eta_\ell^*))\to_p0$ and
\(
\sup_{\eta}\|\tilde H_n^{(\ell)}(\eta)-H_n^{(\ell)}(\eta)\|\to_p0.
\)
Assumption~\ref{ass:diff-info-stage} implies
$H_n^{(\ell)}(\eta)\to_p -I^{(\ell)}(\eta)$ uniformly on
$\mathcal C\times\Pi_{\min}$, so
$\tilde H_n^{(\ell)}(\tilde\eta_{\ell,n})\to_p -I^{(\ell)}(\eta_\ell^*)$.
By Slutsky's theorem,
\[
\sqrt{n}(\hat\eta_{\ell,n}-\eta_\ell^*)
\ \rightsquigarrow\ \mathcal N\big(0,I^{(\ell)}(\eta_\ell^*)^{-1}\big).
\]

\printbibliography

\end{document}